\documentclass[hidelinks,onefignum,onetabnum]{siamart251216}

\usepackage[T1]{fontenc}
\usepackage{amsfonts}
\usepackage{graphicx}
\usepackage{epstopdf}
\usepackage{enumitem}
\ifpdf
  \DeclareGraphicsExtensions{.eps,.pdf,.png,.jpg}
\else
  \DeclareGraphicsExtensions{.eps}
\fi
\setlist[enumerate,1]{label=(\roman*)}
\usepackage{amsmath, amssymb, dsfont}
\usepackage{bbm}
\usepackage{standalone}
\usepackage{mathtools}
\newcommand{\N}{\mathbb{N}}
\newcommand{\Z}{\mathbb{Z}}
\newcommand{\Q}{\mathbb{Q}}
\newcommand{\R}{\mathbb{R}}

\DeclareMathOperator*{\s}{\mathcal S}

\DeclareMathOperator{\conv}{conv}

\DeclareMathOperator{\Eq}{eq}

\newcommand{\spmatrix}[1]{\bigl(\begin{smallmatrix}#1\end{smallmatrix}\bigr)}

\newcommand{\eps}{\varepsilon} 
\renewcommand{\epsilon}{\varepsilon}

\usepackage{xcolor}
\definecolor{grayif}{gray}{0.7}
\definecolor{lightpurple}{RGB}{170,0,255}
\definecolor{myorange}{RGB}{255,140,0} 

\usepackage{xcolor}
\definecolor{grayif}{gray}{0.7}
\definecolor{lightpurple}{RGB}{170,0,255}
\definecolor{myorange}{RGB}{255,140,0} 

\usepackage{tikz}
\usetikzlibrary{arrows.meta,positioning,calc}

\usepackage{pgfplots}

\expandafter\let\csname algorithm*\endcsname\relax
\expandafter\let\csname endalgorithm*\endcsname\relax

\usepackage[ruled,vlined]{algorithm2e}
\usepackage{setspace}

\SetCommentSty{mycommfont}
\newcommand{\vgap}{\vspace{.2em}}
\SetKw{KwAnd}{and}
\SetKw{KwOr}{or}
\SetKw{KwNot}{not}
\DontPrintSemicolon
\SetKw{KwReturn}{return}
\SetKwFor{For}{for}{\!:}{endfor}
\SetKwFor{While}{while}{\!:}{endw}
\SetKwFor{ForEach}{for each}{\!:}{endfor}
\SetKwIF{If}{ElseIf}{Else}{if}{\!:}{else if}{else}{endif}
\SetKwComment{tcp}{(}{)}

\iftrue
\makeatletter
\renewcommand{\SetKwInOut}[2]{%
  \sbox\algocf@inoutbox{\KwSty{#2}\algocf@typo:}%
  \expandafter\ifx\csname InOutSizeDefined\endcsname\relax
    \newcommand\InOutSizeDefined{}\setlength{\inoutsize}{\wd\algocf@inoutbox}%
    \sbox\algocf@inoutbox{\parbox[t]{\inoutsize}{\KwSty{#2}\algocf@typo:\hfill}~}\setlength{\inoutindent}{\wd\algocf@inoutbox}%
  \else
    \ifdim\wd\algocf@inoutbox>\inoutsize%
    \setlength{\inoutsize}{\wd\algocf@inoutbox}%
    \sbox\algocf@inoutbox{\parbox[t]{\inoutsize}{\KwSty{#2}\algocf@typo:\hfill}~}\setlength{\inoutindent}{\wd\algocf@inoutbox}%
    \fi%
  \fi
  \algocf@newcommand{#1}[1]{%
    \ifthenelse{\boolean{algocf@inoutnumbered}}{\relax}{\everypar={\relax}}%
    {\let\\\algocf@newinout\hangindent=\inoutindent\hangafter=1\parbox[t]{\inoutsize}{\KwSty{#2}\algocf@typo:\hfill}~##1\par}%
    \algocf@linesnumbered
  }}%
\makeatother
\fi
\SetKwInOut{Input}{input}
\SetKwInOut{Output}{output}

\usepackage{lscape}  
\usepackage{afterpage}

\usepackage{orcidlink}

\renewcommand{\vec}[1]{\ensuremath{\boldsymbol{#1}}}

\renewcommand{\epsilon}{\varepsilon}

\newcommand{\np}{\ensuremath{\textup{\textsf{NP}}}\xspace}
\newcommand{\cls}{\ensuremath{\textup{\textsf{CLS}}}\xspace}
\newcommand{\ppad}{\ensuremath{\textup{\textsf{PPAD}}}\xspace}
\newcommand{\pls}{\ensuremath{\textup{\textsf{PLS}}}\xspace}

\newsiamremark{remark}{Remark}
\newsiamremark{hypothesis}{Hypothesis}
\crefname{hypothesis}{Hypothesis}{Hypotheses}
\newsiamthm{claim}{Claim}
\newsiamremark{fact}{Fact}
\crefname{fact}{Fact}{Facts}
\newsiamthm{question}{Question}

\headers{An Unconditional Lower Bound for the Active-Set Method}{E. Bach, Y. Disser, S. Huiberts, and N. Mosis}

\title{An Unconditional Lower Bound for the Active-Set Method in Almost-Linear Programming\thanks{A preliminary version of the special case $\eps = 1$ of our result appeared in the proceedings of SODA 2026~\cite{bach2026unconditional}.
}}

\author{Eleon Bach\thanks{TU M\"unchen, Germany (\email{eleon.bach@proton.me}).}
    \and Yann Disser\thanks{TU Darmstadt, Germany (\email{disser@mathematik.tu-darmstadt.de}).~\orcidlink{0000-0002-2085-0454}}  
    \and Sophie Huiberts\thanks{CNRS, Univ Clermont Auvergne, INP Clermont Auvergne, Mines Saint-Etienne, LIMOS, \\ \indent France (\email{sophie@huiberts.me}).~\orcidlink{0000-0003-2633-014X}}
  \and Nils Mosis\thanks{TU Darmstadt, Germany (\email{mosis@mathematik.tu-darmstadt.de}).~\orcidlink{0000-0002-0692-0647}}}

\usepackage{amsopn}

\ifpdf
\hypersetup{
  pdftitle={An Unconditional Lower Bound for the Active-Set Method in Almost-Linear Programming},
  pdfauthor={E. Bach, Y. Disser, S. Huiberts, and N. Mosis}
}
\fi

\begin{document}

\maketitle

\begin{abstract}
We prove that, for every fixed~$\epsilon\in(0,1]$, the active-set method requires exponentially many iterations in the worst case to maximize a convex objective of the form~$(x_1)^{1+\epsilon}+\vec c^{\top}\vec x$ subject to linear constraints, regardless of the pivot rule used.
Our lower bound therefore holds for objectives arbitrarily close to the linear case~$\epsilon=0$, where the active-set method coincides with the simplex method and a lower bound for all pivot rules would constitute a major breakthrough.
The special case~$\eps=1$ yields a lower bound in convex quadratic maximization.
This improves on the best previously known lower bound [IPCO 2025] requiring polynomial objectives of degree~$\omega(\log d)$ in dimension~$d$ and firmly resolves the open question of whether a constant degree suffices.

Our result is based on a novel extended formulation, recursively constructed using deformed products.
Its key feature is that it projects onto a polygonal approximation of a parabola while preserving all of its exponentially many vertices.
We define a convex objective with a single non-linear variable that forces the active-set method to follow the parabolic boundary of this projection, without allowing any shortcuts along chords corresponding to edges of its full-dimensional preimage.
\end{abstract}



\section{Introduction}
A key difficulty in the worst-case analysis of active-set methods lies in ruling out the existence of an efficient pivot rule, rather than eliminating one rule at a time.
We address this fundamental challenge by proving, for every fixed~$0<\epsilon\leq 1$, a pivot-rule-independent and unconditional lower bound for maximizing an almost-linear objective of the form~$(x_1)^{1+\eps}+\vec c^{\top}\vec x$ subject to linear constraints.

Active-set methods approach constrained optimization problems by maintaining a \emph{working set} of constraints, variables, or vertices that determines the reduced subproblem or search directions at each iterate and is updated using feasibility and optimality information.
The canonical example is the simplex method, where, under nondegeneracy, the tight constraints determine the current vertex, reduced costs identify the improving incident edges, and moving to an adjacent vertex exchanges one tight constraint for another~\cite{dantzig1963linear}.
Beyond linear programming, the same principle appears in gradient-projection methods for box-constrained optimization, where the active bounds define the current face~\cite{hager2006new,more1991solution}; in primal--dual active-set methods for quadratic programs and complementarity problems, where a set of constraints is treated as active and updated using primal and multiplier information~\cite{curtis2016globally,hintermuller2002primal,hungerlander2015feasible}; in away-step Frank--Wolfe methods for smooth optimization over polytopes~\cite{bomze2020active,pena2016neumann}, or Wolfe's algorithm for finding the minimum-norm point of a polytope~\cite{deLoera2020minimum,wolfe1976finding}, where the current iterate is represented as a convex combination of selected vertices; 
and in active-set methods for support-vector-machine training, where selected dual variables are fixed at their bounds~\cite{scheinberg2006efficient}.
While some of these methods update the working set by adding or removing multiple elements simultaneously~\cite{curtis2016globally,hager2006new,hungerlander2015feasible}, we focus on single-element updates.
If there is a choice among elements that can enter or leave, or among search directions, we refer to a rule that resolves such choices as a \emph{pivot rule}.

In the present article, we consider the active-set method for finding a Karush--Kuhn--Tucker (KKT) point of a convex maximization program of the form
\begin{equation*}\tag{\ensuremath{\mathrm{LP}_{\varepsilon}}}\label{LP epsilon}
\max ~(x_1)^{1+\eps}+\vec c^{\top}\vec x \text{ subject to } A\vec{x}\le \vec{b},~x_1\geq0,
\end{equation*}
where~$\eps\in[0,1]$,~$\vec c\in\Q^d$,~$A\in\Q^{m\times d}$,~$\vec b\in\Q^m$.

In the quadratic case~$\eps=1$, computing a global optimum of~\eqref{LP epsilon} is \np-hard~\cite {pardalos1991quadratic}.
Recently, Fearnley et al.~\cite{fearnley2025complexity} showed that the simpler task of finding a KKT point remains computationally hard in quadratic programming, by proving \cls-completeness, where~$\cls=\ppad\cap\pls$~\cite{fearnley2022complexity}, of the problem.
As far as we know, however, it is open whether this hardness persists for objectives of the form~$(x_1)^2+\vec{c}^{\top}\vec{x}$.
More generally, we are not aware of any results on the complexity of finding a KKT point for~\eqref{LP epsilon} with~$\eps\in(0,1]$.

In the linear case~$\eps=0$, every KKT point is a global optimum, and the ellipsoid method~\cite{Khachiyan79,kozlov1979polynomial} and interior-point methods~\cite{karmarkar1984new,ye1989extension} solve the problem in polynomial time.
Whether a \emph{strongly} polynomial-time algorithm exists remains a major open problem and constitutes Smale's 9th problem for the 21st century~\cite{smale1998mathematical}.
A prominent candidate for yielding such an algorithm is the simplex method.
Whether the method admits an efficient pivot rule has been a notorious open question since its inception by Dantzig in~1947~\cite{Dantzig82}, and could yield a strongly polynomial algorithm for linear optimization, as well as a proof of the polynomial (monotone) Hirsch conjecture~\cite{dantzig1963linear}. 
These are widely regarded as two of the most important open problems in the theory of linear programming~\cite{santos2012counterexample,smale1998mathematical}.

Exponential worst-case constructions for the simplex method were found early on for many natural pivot rules~\cite{avis1978notes, goldfarb1979worst, jeroslow1973simplex, klee1972good, murty1980computational}, using deformations of the hypercube, which were later unified by Amenta and Ziegler~\cite{amenta1999deformed} in terms of~\emph{deformed products} (see Section~\ref{sec:extended_form}).
It took a breakthrough several decades later by Friedmann in 2011~\cite{friedmann2011subexponential_1}, based on a connection to parity games and Markov decision processes, in order to tackle more balanced pivot rules, such as randomized pivot rules~\cite{friedmann2011subexponential_2} and history-based pivot rules~\cite{avis2017exponential, disser2023exponential}.
The majority of existing results are tailored to individual rules,
with exceptions including recent lower bounds by Black~\cite{Black27,black2025exponentiallowerbounds} for all shadow vertex rules and some normalized-weight rules~\cite{black2023normalizedweight}, as well as constructions that apply to a small number of rules simultaneously~\cite{disser2026lower, disser2018simplex}.
It seems highly unlikely that any of these constructions can be generalized to \emph{unconditional} lower bounds \emph{for all pivot rules}.

Disser and Mosis~\cite{disser2025unconditional} recently obtained an unconditional lower bound for the active-set method with single-element updates for smooth optimization over polytopes (see Section~\ref{sec: simplex and active set}).
They proved an exponential lower bound on its running time that holds \emph{for all pivot rules} --- the first such result, in any setting, that is unconditional~(to the best of our knowledge).
For obtaining a super-polynomial lower bound, they required polynomial objectives of degree~$\omega(\log d)$ in dimension~$d$.

In this paper, we substantially strengthen their result by showing that an unconditional exponential lower bound already holds for the convex, almost-linear objectives in~$\eqref{LP epsilon}$.
Our proof is based on an extended formulation~\cite{Kaibel11,Yannakakis91} constructed using the deformed-product framework of Amenta and Ziegler~\cite{amenta1999deformed}.
Our extended formulation, when projected appropriately, approximates a portion of the parabola to high accuracy, and may be of independent interest.

\subsection*{Our results}

We prove an unconditional exponential lower bound on the number of iterations of the active-set method for finding a KKT point of~\eqref{LP epsilon}, for every fixed~$\eps\in(0,1]$, that holds \emph{for all pivot rules}.

\begin{theorem}[Theorem~\ref{thm:main_details epsilon}]\label{thm:main epsilon}
For every even~$d\geq 2$, there exists a polytope~$\mathcal{P} \subseteq {\R}^d\cap\{\vec x\colon x_1>0\}$ with~$2d$ facets and, for every fixed~$\eps\in(0,1]$, a convex function~$f_{\eps}\colon \R^d \to \R$ of the form~$f_{\eps}(\vec x)=\vert x_1\vert^{1+\eps}+\vec {c_\eps}^{\top}\vec x$, where~$\vec {c_\eps}\in\R^d$ (with~$\vec{c_\eps}\in\Q^d$ whenever~$\eps\in\Q$), such that, for some starting vertex, the active-set method needs~$2^d-1$ iterations to find a KKT point of~$\max\{f_\eps(\vec x)\colon\vec x\in\mathcal{P}\}$, irrespective of the pivot rule.
\end{theorem}

A similar result for the linear case~$\eps=0$, where the active-set method coincides with the simplex method (see Section~\ref{sec: simplex and active set}), would disprove the existence of a polynomial-time pivot rule for the simplex method, which would constitute a major breakthrough.

In the quadratic case~$\eps=1$, Theorem~\ref{thm:main epsilon} yields the following.

\begin{corollary}
    The active-set method needs an exponential number of iterations in the worst case for finding a KKT point of a convex quadratic maximization program with linear constraints, irrespective of the pivot rule.
\end{corollary}

This corollary substantially improves on the lower bound by Disser and Mosis~\cite{disser2025unconditional}, which requires a polynomial degree of~$\omega(\log d)$ for a super-polynomial bound and a polynomial degree of~$\Omega(d)$ for an exponential bound, to an exponential lower bound for a convex polynomial of degree~$2$, and it firmly resolves the open question (Question~2 in~\cite{disser2025unconditional}) of whether the active-set method yields a polynomial-time algorithm for any constant degree.

The fact that our lower bounds hold for \emph{convex maximization} objectives is a qualitative improvement in terms of approaching a result for the simplex method, in the sense that convex functions over convex, compact domains attain their maxima at extreme points.
This means that the active-set method (when started at a vertex) moves along edges of the polyhedron until reaching a KKT point at a vertex, in simplex-like fashion.
For non-convex objectives, such as the one used by Disser and Mosis~\cite{disser2025unconditional}, the active-set method generally moves through the interior --- a behavior that differs qualitatively from the simplex method.

A recent line of work approaches the complexity of the simplex method by showing that it is \np-hard to approximate the shortest monotone path between vertices of a polytope~\cite{cardinal2025inapproximability, deloera2022pivot}.
In contrast, the value of our contribution lies in the fact that we provide an \emph{unconditional} lower bound.
In particular, the complexity of finding a KKT point for~\eqref{LP epsilon} with~$\eps\in(0,1]$ remains unresolved.

Our proof is based on a novel extended formulation~\cite{Kaibel11,Yannakakis91} using deformed products~\cite{amenta1999deformed} that may be interesting in its own right. 
To the best of our knowledge, this represents the first extended formulation constructed via deformed products.
\begin{theorem}[Theorem~\ref{thm:extended parabola}]\label{thm:main2}
    For every~$n,\, d \in \mathbb N$ such that $d$ and~$n/(2d)$ are even integers, there exists a polytope~$\mathcal{P} \subseteq \R^d$ with~$n/2$ facets and~$(n/d)^{d/2}$ vertices and an orthogonal projection $\R^d \to \R^2$ such that the vertices project onto
    \[
      \left\{\begin{pmatrix} x \\x^2 - x \end{pmatrix} \in \R^2 \colon  x \in [0,1], \ x\cdot\left((n/d)^{d/2}-1\right) \in \mathbb{Z}\right\}.
    \]
\end{theorem}

The key feature of our extended formulation is that its projection yields a poly\-gonal approximation of a parabola to any accuracy.
Every vertex of the corresponding polyhedron, which is combinatorially equivalent to a cube if~$n=4d$, projects to a unique vertex of the two-dimensional polygon.
In particular, it has exponentially many vertices which are preserved under the projection, i.e., no vertex projects into the interior of the shadow and no two vertices get projected to the same vertex.
This property turns out to be crucial for our proof and does not hold for 
previous extended formulations, such as the one of Ben-Tal and Nemirovski~\cite{ben2001polyhedral} yielding a regular $n$-gon (see Figure~\ref{fig:Ben-Tal--Nemirovski}). 

\subsection*{Further related work}

A topic closely related to (constraint-based) active-set methods is \emph{active-set identification}, which studies when the constraints tight at a solution can be correctly identified from nearby iterates~\cite{burke1994exposing,Facchinei1998accurate,oberlin2006active}. 
Correct identification exposes the relevant face and locally reduces the inequality-constrained problem to an equality-constrained one.
More broadly, analogous identification phenomena arise in nonsmooth optimization, where the relevant local structure can be an active manifold on which the objective is smooth~\cite{lewis2002active}.

Extended formulations were introduced in the seminal work of Yannakakis~\cite{Yannakakis91}.
They have since become a standard tool in combinatorial optimization, with applications including~$0$--$1$ optimization~\cite{lovasz1991cones}, the Steiner traveling salesperson problem~\cite{baiou2002steiner}, and network design~\cite{mahjoub2019distance}.
One of their most popular practical applications is the approximation of the second-order cone using extended formulations of regular polygons.
The first such construction was given by Ben-Tal and Nemirovski~\cite{ben2001polyhedral}.
This construction was slightly improved and generalized later~\cite{fiorini2012extended,glineur,kaibel2013constructing,VANDAELE2017217}.
In particular, Kaibel and Pashkovich~\cite{kaibel2013constructing} introduced a general framework for constructing extended formulations based on symmetries called \emph{reflection relations}.

The open-source solver SCIP~\cite{scip7} included the extended formulations by Ben-Tal and Nemirovski~\cite{ben2001polyhedral} and by Glineur~\cite{glineur} up to version 7.
This allowed the construction of linear approximations that allow LP algorithms, such as the simplex method, to handle second-order conic constraints.
Although they are no longer included in the solver as of version 8, a number of works reported strong computational results using these constructions compared to the alternative approach of gradient-based separation, particularly when applied to robust or portfolio optimization problems~\cite{barmann2016polyhedral,Vielma2008}, which suggests a practical relevance of these constructions.

Variations of the construction by Ben-Tal and Nemirovski have also been used to linearly approximate the quadratic costs of the unit commitment problem in electrical power production. 
In this context, Angulo Cárdenas et al.~\cite{cardenas2016} introduced an extended formulation approximating a quadratic curve using an approach comparable to the reflection relations of~\cite{fiorini2012extended}.
Another variation on the construction of Ben-Tal and Nemirovski was used by Huiberts et al.~\cite{huiberts2023upper} to prove a lower bound on the smoothed complexity of the simplex method using the shadow vertex pivot rule. 

\begin{figure}
\centering
\begin{tikzpicture}
\begin{axis}[width=8cm, height=8cm,
  xlabel={$x_1$}, ylabel={$x_2$},
  xlabel style={at={(axis description cs:0.5,0.05)}, anchor=north},
  ylabel style={at={(axis description cs:0.15,0.5)}, anchor=south},
  xtick={-1,1}, ytick={-1,1},
  grid=both, axis equal image, axis on top, clip=false]

\definecolor{lightpurple}{RGB}{170,0,255}
\definecolor{lightpurplefill}{RGB}{230,200,255}
\definecolor{myorange}{RGB}{255,140,0}
\addplot[fill=lightpurplefill, draw=lightpurple, opacity=0.3, thick] coordinates {
  (-1.0196,0.0000)
  (-0.9420,-0.3902)
  (-0.7210,-0.7210)
  (-0.3902,-0.9420)
  (0.0000,-1.0196)
  (0.3902,-0.9420)
  (0.7210,-0.7210)
  (0.9420,-0.3902)
  (1.0196,0.0000)
  (0.9420,0.3902)
  (0.7210,0.7210)
  (0.3902,0.9420)
  (0.0000,1.0196)
  (-0.3902,0.9420)
  (-0.7210,0.7210)
  (-0.9420,0.3902)
  (-1.0196,0.0000)
};
\addplot[only marks, color=myorange, mark=square*, mark size=2pt] coordinates {(0.0000,0.0000)};
\addplot[only marks, color=lightpurple, mark=*, mark size=2pt] coordinates {(0.0000,1.0196)};
\addplot[only marks, color=lightpurple, mark=*, mark size=2pt] coordinates {(0.0000,-1.0196)};
\addplot[only marks, color=lightpurple, mark=*, mark size=2pt] coordinates {(1.0196,0.0000)};
\addplot[only marks, color=lightpurple, mark=*, mark size=2pt] coordinates {(0.7210,0.7210)};
\addplot[only marks, color=lightpurple, mark=*, mark size=2pt] coordinates {(0.9420,0.3902)};
\addplot[only marks, color=myorange, mark=square*, mark size=2pt] coordinates {(0.8315,0.5556)};
\addplot[only marks, color=myorange, mark=square*, mark size=2pt] coordinates {(0.9808,0.1951)};
\addplot[only marks, color=lightpurple, mark=*, mark size=2pt] coordinates {(0.3902,0.9420)};
\addplot[only marks, color=myorange, mark=square*, mark size=2pt] coordinates {(0.5556,0.8315)};
\addplot[only marks, color=myorange, mark=square*, mark size=2pt] coordinates {(0.1951,0.9808)};
\addplot[only marks, color=lightpurple, mark=*, mark size=2pt] coordinates {(0.7210,-0.7210)};
\addplot[only marks, color=lightpurple, mark=*, mark size=2pt] coordinates {(0.9420,-0.3902)};
\addplot[only marks, color=myorange, mark=square*, mark size=2pt] coordinates {(0.8315,-0.5556)};
\addplot[only marks, color=myorange, mark=square*, mark size=2pt] coordinates {(0.9808,-0.1951)};
\addplot[only marks, color=lightpurple, mark=*, mark size=2pt] coordinates {(0.3902,-0.9420)};
\addplot[only marks, color=myorange, mark=square*, mark size=2pt] coordinates {(0.5556,-0.8315)};
\addplot[only marks, color=myorange, mark=square*, mark size=2pt] coordinates {(0.1951,-0.9808)};
\addplot[only marks, color=lightpurple, mark=*, mark size=2pt] coordinates {(-1.0196,0.0000)};
\addplot[only marks, color=lightpurple, mark=*, mark size=2pt] coordinates {(-0.7210,0.7210)};
\addplot[only marks, color=lightpurple, mark=*, mark size=2pt] coordinates {(-0.9420,0.3902)};
\addplot[only marks, color=myorange, mark=square*, mark size=2pt] coordinates {(-0.8315,0.5556)};
\addplot[only marks, color=myorange, mark=square*, mark size=2pt] coordinates {(-0.9808,0.1951)};
\addplot[only marks, color=lightpurple, mark=*, mark size=2pt] coordinates {(-0.3902,0.9420)};
\addplot[only marks, color=myorange, mark=square*, mark size=2pt] coordinates {(-0.5556,0.8315)};
\addplot[only marks, color=myorange, mark=square*, mark size=2pt] coordinates {(-0.1951,0.9808)};
\addplot[only marks, color=lightpurple, mark=*, mark size=2pt] coordinates {(-0.7210,-0.7210)};
\addplot[only marks, color=lightpurple, mark=*, mark size=2pt] coordinates {(-0.9420,-0.3902)};
\addplot[only marks, color=myorange, mark=square*, mark size=2pt] coordinates {(-0.8315,-0.5556)};
\addplot[only marks, color=myorange, mark=square*, mark size=2pt] coordinates {(-0.9808,-0.1951)};
\addplot[only marks, color=lightpurple, mark=*, mark size=2pt] coordinates {(-0.3902,-0.9420)};
\addplot[only marks, color=myorange, mark=square*, mark size=2pt] coordinates {(-0.5556,-0.8315)};
\addplot[only marks, color=myorange, mark=square*, mark size=2pt] coordinates {(-0.1951,-0.9808)};
\end{axis}
\end{tikzpicture}
\caption{Projection of the intersection of the Ben-Tal--Nemirovski~\cite{ben2001polyhedral} polyhedral approximation of the Lorentz cone $\{(t,x) : \|x\| \le t\}$ with the halfspace $t \leq 1$, projected to the plane $t=1$ (purple). Additional vertices (orange squares) are not preserved, i.e., not projected onto vertices.\label{fig:Ben-Tal--Nemirovski} }
\end{figure}

\section{Simplex and active-set method}\label{sec: simplex and active set}

The active-set method is a well-known and natural generalization of the simplex method to nonlinear objectives (see~Exercise~8.17 in~\cite{fletcher2000practical}, \cite{disser2025unconditional}).
We briefly outline both algorithms.

~ 

\noindent\textbf{The simplex method} solves linear programs of the form
\begin{equation}\tag{LP}\label{LP}
\begin{aligned}
    \max \vec{c}^{\top}\vec{x} \text{ subject to }  A\vec{x} \leq \vec{b},
\end{aligned}
\end{equation}
where~$\vec c\in\Q^d$,~$A\in\Q^{m\times d}$, and~$\vec b\in\Q^m$.
The method can be understood combinatorially as a traversal of~$\mathcal{P}=\{\vec x\colon A\vec x\leq \vec b\}$ along its edges, such that the objective is monotonically increasing.
The number of iterations it takes depends on the \emph{pivot rule} that governs which improving edge direction is chosen in each step in case there is more than one.
Whether there exists a pivot rule guaranteeing a polynomial running time of the simplex method is arguably one of the most famous open problems in (linear) optimization.
We refer to~\cite{dantzig1963linear} for more details on the simplex method.

~

\noindent\textbf{The active-set method} can be applied to nonlinear programs of the form
\begin{equation}\tag{NLP}\label{NLP}
\begin{aligned}
    \max f(\vec x) \text{ subject to } A\vec{x} \leq \vec{b},
\end{aligned}
\end{equation}
where~$f\colon\R^d\to\R$ is continuously differentiable,~$A\in\Q^{m\times d}$, and~$\vec b\in\Q^m$.
As before, we denote the feasible region by~$\mathcal{P}$. 
The set of indices of the constraints that are active at some~$\vec{x}\in\mathcal{P}$ is denoted by~$\Eq(\vec x)$.
For simplicity, we assume that~$\mathcal{P}$ is bounded.

We use a formulation of the active-set method by Disser and Mosis~\cite{disser2025unconditional} (see Algorithm~\ref{alg:AS}), which they proved to be a natural generalization of the simplex method and the active-set method for concave QPs to non-concave maximization programs.

\begin{figure}[htb]
  \centering
    \begin{algorithm}[H]
    \begin{minipage}{\dimexpr\textwidth-21pt} 
        \setstretch{1.2} 
     
	    \label{alg:AS}\caption{\textsc{ActiveSet}($f$, $A$, $\vec b$, $\vec{x}$)}
	    \textbf{input: }continuously differentiable~$f\colon\R^d\to\R$, matrix~$A\in\Q^{m\times d}$,~$\vec{b}\in\Q^m$,
	    \Indp point~$\vec{x}$ of polytope~$\mathcal P\coloneqq\{\vec y\colon A\vec y\leq \vec b\}$\\
	    \Indm \textbf{output: }Karush--Kuhn--Tucker (KKT) point~$\vec x$ of~$\max\{f(\vec y)\colon\vec y\in\mathcal{P}\}$
	    \vgap
	    \hrule  
	    \vgap
	    
	    $\mathcal{A}\gets\Eq(\vec x)$\tcp*{active set} 
        \SetInd{4pt}{4pt} 
	    
	    \While{$\mathcal{D}\coloneqq\{\vec{d}\colon{A_{\Eq(\vec x)\cdot}}\vec d\leq \vec 0,~\vec{\nabla\!f(\vec{x})}\!^{\top}\vec{d}> 0\}\neq\emptyset$}
	    {
	    \vspace{2pt}
		 take $\vec{ d}\in\mathcal{D}$ maximizing~$|\{i\in\mathcal{A}\colon \vec{A_{i\cdot}}\vec{ d}=0\}|$\; 	     
	     
	     \If{$A_{\mathcal{A}\cdot}\vec{d}\neq\vec 0$}{
	     $\mathcal{A}\gets\mathcal{A}\setminus\{i\}$ for some~$i\in\mathcal{A}$ with~$\vec{A_{i\cdot}}\vec{d}<0$\;	
		 
	     }{}
	     \If{$A_{\mathcal{A}\cdot}\vec{d}=\vec 0$}{	     
	     ${\mu}\gets\inf\{\mu\geq 0\colon \vec{x}+\mu\vec{d}\notin\mathcal{P}\text{ or } \vec{\nabla\!f(x}+\mu \vec{d)}\!^{\top}\vec d\leq 0\}$\;
	     
	     $\vec{x}\gets\vec{x}+\mu\vec{d}$\;
	     
	     \If{$\vec{\nabla\!f(\vec{x})}\!^{\top}\vec{d}> 0$}
	     {
	     $\mathcal{A} \gets \mathcal{A}\cup\{j\}$ for some~$j\in\Eq(\vec{x})\setminus\mathcal{A}$\;
	     }    	 	
	     }    
	    }
    \Return \vec x
    \vspace{2pt}
    \end{minipage}
	\end{algorithm}
\end{figure}

The running time of the algorithm, i.e., the number of iterations, again highly depends on the \emph{pivot rule} that determines the direction~$\vec{d}$ and the indices~$i\in\mathcal{A}$ and~$j\in\Eq(\vec{x})\setminus\mathcal{A}$ in each iteration in case there is more than one choice.
Disser and Mosis~\cite{disser2025unconditional} proved that \textsc{ActiveSet} does not admit any polynomial-time pivot rule for objective functions that are polynomials of degree~$\omega(\log d)$, but the existence of efficient pivot rules was open for lower polynomial degrees.
We improve this result to an exponential lower bound for maximizing convex objectives
of the form~$(x_1)^{1+\eps}+\vec c^{\top}\vec x$, for every fixed~$\eps\in(0,1]$;
for~$\eps=1$, this is a convex quadratic polynomial.

\section{An unconditional exponential lower bound for active-set}\label{sec: exponential bound}

The starting point of our lower bound construction is the observation that it is easy to force the active-set method (or the simplex method) to visit roughly half of the vertices of a polygon in~$\R^2$ via a suitably chosen linear objective.
However, this only yields a trivial linear lower bound on the number of iterations, since the polygon is encoded by a number of inequalities, corresponding to facets, which is linear in the number of vertices.
This approach clearly is not viable for our purposes.
Instead we resort to an extended formulation of a polygon with many vertices using a polyhedron with more dimensions but exponentially fewer constraints.

Using an extended formulation, however, poses the challenge that we introduce many additional edges between vertices that are not adjacent in the projection.
In particular, the active-set method (or simplex) is no longer constrained to move along the boundary of the projected polygon, and may, depending on the choice of the pivot rule, even move to vertices projecting to its interior, which makes it much harder to control adversarially.
Known extended formulations, such as the one of Ben-Tal and Nemirovski~\cite{ben2001polyhedral} do not preserve all vertices under projection (see Figure~\ref{fig:Ben-Tal--Nemirovski}). 

We overcome this issue by carefully designing a novel extended formulation whose vertices project onto a suitable parabola (Section~\ref{sec:extended_form}).
By convexity of the parabola, this means in particular that none of its vertices are projected to the interior of the resulting polygon.

In addition to constructing a suitable extended formulation, we need to choose an objective function that prevents improving moves along chords\footnote{A \emph{chord} is a line segment connecting two vertices of a polygon that otherwise lies entirely in its interior.} of the projected polygon, which may correspond to edges of the preimage (i.e., the extended formulation). 
We construct a convex maximization objective with this property in~\Cref{sec:lower_bound}.

\subsection{An extended formulation approximating a parabola\label{sec:extended_form}}

We construct a convex polyhedron that has a two-dimensional projection
closely approximating the set $\conv\{(x,\, x^2-x) : x \in [0,1]\}$.
In contrast to previous constructions via \emph{reflection relations} introduced by Kaibel and Pashkovich~\cite{kaibel2013constructing}, such as the one by Huiberts, Lee and Zhang~\cite{huiberts2023upper}, our construction is based on the notion of \emph{deformed products} as introduced by Amenta and Ziegler~\cite{amenta1999deformed}.

This operation requires polyhedra that are compatible in the following sense.

\begin{definition}\label{def:comb equivalence}
    Two polytopes~$\mathcal{P}$, $\mathcal{Q}$ are \emph{combinatorially equivalent} if there is a one-to-one correspondence between their vertex sets~$\{\vec{p^{(1)}},\dots,\vec{p^{(m)}}\} = \operatorname{vert}(\mathcal{P})$ and~$\{\vec{q^{(1)}},\dots,\vec{q^{(m)}}\} = \operatorname{vert}(\mathcal{Q})$ such that the convex hull~$\conv\{\vec{p^{(i)}}\! :\! i\! \in\! I\}$ is a face of~$\mathcal{P}$ if and only if~$\conv\{\vec{q^{(i)}} : i \in I\}$ is a face of~$\mathcal{Q}$ for every~$I \subseteq \{1,\dots,m\}$.
\end{definition}

\begin{definition}    
    Two combinatorially equivalent polytopes are \emph{normally equivalent} if the facet normals of corresponding facets coincide.
\end{definition}

We use the following characterization of deformed products and refer to~\cite{amenta1999deformed} for a detailed introduction.
Note that, for the first item, we assume that the vertices of $\mathcal{V}$ and $\mathcal{W}$ are labeled as per the bijection in Definition~\ref{def:comb equivalence}.

\begin{theorem}[Theorem 3.4 in \cite{amenta1999deformed}]\label{lem:AZ}
    Let $\mathcal{P} \subseteq \R^d$ be a $d$-dimensional polytope, let~$\varphi\colon \R^d \to \R$ be a linear function
    with $\varphi(\mathcal{P}) = [0,1]$,
    and let $\mathcal{V},\, \mathcal{W} \subseteq \R^r$
    be normally equivalent $r$-dimensional polytopes.
    \begin{enumerate}[leftmargin=1.5em]
        \item\label{lem:AZ_item1} If $\mathcal{P}$, $\mathcal{V}$, and $\mathcal{W}$ are given by~$\mathcal{P} = \conv\{\vec{p^{(1)}},\dots,\vec{p^{(m)}}\}$,~$\mathcal{V} = \conv\{\vec{v^{(1)}},\dots,\vec{v^{(n)}}\}$, and~$\mathcal{W} = \conv\{\vec{w^{(1)}},\dots,\vec{w^{(n)}}\}$, the deformed product is
        \begin{equation}
            (\mathcal{P},\varphi) \bowtie (\mathcal{V}, \mathcal{W}) =
        \conv\left\{\begin{pmatrix}\vec{p^{(i)}} \\ \vec{v^{(j)}} + \varphi(\vec{p^{(i)}})(\vec{w^{(j)}} - \vec{v^{(j)}})\end{pmatrix} : \begin{array}{l}i \in [m],\\ j \in [n]\end{array}\right\},
        \end{equation}
        where we write~$[k]\coloneqq\{1,2,\ldots,k\}$, for all~$k\in\N$.
        
        \item\label{lem:AZ_item2} If $\mathcal{P}$, $\mathcal{V}$ and $\mathcal{W}$ are given by
        \begin{equation*}
            \mathcal{P} = \{\vec{x} \in \R^d : A \vec{x} \leq \vec{\alpha}\}, \quad
            \mathcal{V} = \{\vec{u} \in \R^r : B \vec{u} \leq \vec{\beta}\}, \quad 
            \mathcal{W} = \{\vec{u} \in \R^r : B \vec{u} \leq             \vec{\beta'}\},
        \end{equation*}
        the deformed product is
        \begin{equation}
            (\mathcal{P},\varphi) \bowtie (\mathcal{V}, \mathcal{W}) =
            \left\{
            \begin{pmatrix} \vec{x} \\ \vec{u} \end{pmatrix} \in \R^{d+r}
            : A\vec{x} \leq \vec{\alpha}, \, (\vec{\beta} - \vec{\beta'})\,\varphi(\vec{x}) + B\vec{u} \leq \vec{\beta}
            \right\}.
        \end{equation}
        
        \item\label{lem:AZ_item3} The deformed product $(\mathcal{P},\varphi) \bowtie (\mathcal{V}, \mathcal{W})$ is combinatorially equivalent to $\mathcal{P} \times \mathcal{V}$ and to $\mathcal{P} \times \mathcal{W}$.
    \end{enumerate}
\end{theorem}

\begin{figure}
\centering
\begin{tikzpicture}
  \begin{axis}[
    width=12cm,
    height=6cm,
    xlabel={$x_1$},
    ylabel={$x_2$},
    xlabel style={at={(axis description cs:0.5,0.1)}, anchor=north},
ylabel style={at={(axis description cs:0.1,0.5)}, anchor=south},
    xtick={0,1},
    xticklabels={$0$,$1$},
    ytick={-0.25,0},
    yticklabels={$-0.25$, $0$},
    grid=both,
    legend pos=south east,
    legend style={xshift=9pt},
    axis on top,
    clip=false,
  ]

    \addplot [domain=0:1, samples=100, thin, black] {x^2 - x};

    \addplot [color=lightpurple, mark=*, very thick]
      coordinates {
        (0/79, 0)
        (19/79, {(19/79)^2 - (19/79)})
        (20/79, {(20/79)^2 - (20/79)})
        (39/79, {(39/79)^2 - (39/79)})
        (40/79, {(40/79)^2 - (40/79)})
        (59/79, {(59/79)^2 - (59/79)})
        (60/79, {(60/79)^2 - (60/79)})
        (79/79, 0)
      } -- cycle;

    \addplot [color=myorange, mark=square*, very thick]
      coordinates {
        (9/79, {(9/79)^2 - (9/79)})
        (10/79, {(10/79)^2 - (10/79)})
        (29/79, {(29/79)^2 - (29/79)})
        (30/79, {(30/79)^2 - (30/79)})
        (49/79, {(49/79)^2 - (49/79)})
        (50/79, {(50/79)^2 - (50/79)})
        (69/79, {(69/79)^2 - (69/79)})
        (70/79, {(70/79)^2 - (70/79)})
      } -- cycle;

    \legend{$x^2 - x$, $\mathcal{V}_{10,8}$, $\mathcal{W}_{10,8}$}

    \node[lightpurple, font=\small, anchor=south east, xshift=0pt, yshift=-3pt] at (axis cs:0,0) {$v_{0,0}$};
    \node[lightpurple, font=\small, anchor=east, xshift=-2pt, yshift=6pt] at (axis cs:19/79,{(19/79)^2 - (19/79)-0.015}) {$v_{0,1}$};
    \node[lightpurple, font=\small, anchor=east] at (axis cs:20/79,{(20/79)^2 - (20/79)-0.015}) {$v_{1,0}$};
    \node[lightpurple, font=\small, anchor=north east, xshift=2pt, yshift=8pt] at (axis cs:39/79,{(39/79)^2 - (39/79)-0.015}) {$v_{1,1}$};
    \node[lightpurple, font=\small, anchor=north west, xshift=0pt, yshift=8pt] at (axis cs:40/79,{(40/79)^2 - (40/79)-0.015}) {$v_{2,0}$};
    \node[lightpurple, font=\small, anchor=west] at (axis cs:59/79,{(59/79)^2 - (59/79)-0.015}) {$v_{2,1}$};
    \node[lightpurple, font=\small, anchor=west, xshift=2pt, yshift=6pt] at (axis cs:60/79,{(60/79)^2 - (60/79)-0.015}) {$v_{3,0}$};
    \node[lightpurple, font=\small, anchor=south west, xshift=0pt, yshift=-3pt] at (axis cs:1,0) {$v_{3,1}$};

    \node[myorange, font=\small, anchor=east, xshift=-2pt, yshift=6pt] at (axis cs:9/79,{(9/79)^2 - (9/79)-0.01}) {$w_{0,0}$};
    \node[myorange, font=\small, anchor=east] at (axis cs:10/79,{(10/79)^2 - (10/79)-0.01}) {$w_{0,1}$};
    \node[myorange, font=\small, anchor=north east, xshift=-2pt, yshift=9pt] at (axis cs:29/79,{(29/79)^2 - (29/79)-0.01}) {$w_{1,0}$};
    \node[myorange, font=\small, anchor=north east, xshift=4pt, yshift=2pt] at (axis cs:30/79,{(30/79)^2 - (30/79)-0.01}) {$w_{1,1}$};
    \node[myorange, font=\small, anchor=north west, xshift=-4pt, yshift=2pt] at (axis cs:49/79,{(49/79)^2 - (49/79)-0.01}) {$w_{2,0}$};
    \node[myorange, font=\small, anchor=north west, xshift=2pt, yshift=9pt] at (axis cs:50/79,{(50/79)^2 - (50/79)-0.01}) {$w_{2,1}$};
    \node[myorange, font=\small, anchor=west] at (axis cs:69/79,{(69/79)^2 - (69/79)-0.01}) {$w_{3,0}$};
    \node[myorange, font=\small, anchor=west, xshift=2pt, yshift=6pt] at (axis cs:70/79,{(70/79)^2 - (70/79)-0.01}) {$w_{3,1}$};

  \end{axis}
\end{tikzpicture}
\caption{Polygons $\mathcal{V}_{M,N}$ and $\mathcal{W}_{M,N}$ for $M=10$, $N=8$.\label{fig:V_and_W}}
\end{figure}

Our construction is based on two polygons, each with~$N$ mutually distinct vertices of the form~$\vec{h}(x) \coloneqq \spmatrix{x \\ x^2 - x}$ for~$x \in [0,1]$.
For even~$N \in \N$ and some $M \in \N$, we define the polygons (see Figure~\ref{fig:V_and_W})
\begin{align*}
  \mathcal{V}_{M,N} &\coloneqq \conv\{\vec{v_{j,\ell}} : j \in \{0,\dots, N/2-1\},\ell \in \{0,1\}\},\\
  \mathcal{W}_{M,N} &\coloneqq \conv\{\vec{w_{j,\ell}} : j \in \{0,\dots, N/2-1\},\ell \in \{0,1\}\},
\end{align*}
with
\begin{align*}
  \vec{v_{j,\ell}} \equiv \vec{v_{j,\ell}}(M,N) &\coloneqq \vec{h}\!\left(\frac{2M(j + \ell) - \ell}{MN - 1}\right),\\
  \vec{w_{j,\ell}} \equiv \vec{w_{j,\ell}}(M,N) &\coloneqq \vec{h}\!\left(\frac{M(2j + 1) - (1 - \ell)}{MN - 1}\right).
\end{align*}

Observe that, since all vertices of~$\mathcal{V}_{M,N}$ and~$\mathcal{W}_{M,N}$ are distinct and lie on the convex curve $x \mapsto x^2-x$, the polygon~$\conv(\mathcal{V}_{M,N}\cup \mathcal{W}_{M,N})$ has exactly $2N$ vertices.
Furthermore, we show that~$\mathcal{V}_{M,N}$ and~$\mathcal{W}_{M,N}$ are suited for taking deformed products.

\begin{lemma}
    The polygons~$\mathcal{V}_{M,N}$ and~$\mathcal{W}_{M,N}$ are normally equivalent for every even~$M \geq 2$ and~$N \geq 4$.
\end{lemma}
\begin{proof}
    Since the vertices of~$\mathcal{V}_{M,N}$ and~$\mathcal{W}_{M,N}$ lie on the convex curve $x \mapsto x^2-x$, they can be ordered along their respective boundaries by increasing first coordinates, where~$(\vec{v_{j,0}})_1 < (\vec{w_{j,0}})_1 < (\vec{w_{j,1}})_1 < (\vec{v_{j,1}})_1 < (\vec{v_{j+1,0}})_1$.
    We obtain the bijection mapping~$\vec{v_{j,\ell}}$ to $\vec{w_{j,\ell}}$, which trivially preserves the faces (i.e., vertices and edges) of the polygons.

    To prove that the polygons are normally equivalent, we need to verify that the slopes of their edges coincide.
    The slope of an edge between vertices $\vec u$ and $\vec v$ is given by $\operatorname{slope}(\vec{u}; \vec{v}) \coloneqq (v_2-u_2)/(v_1-u_1)$.
    For two vertices of the form~$\vec{h}(x)$ and~$\vec{h}(y)$ with~$x,y \in \R$, we obtain 
    \begin{equation*}
        \operatorname{slope}(\vec{h}(x); \vec{h}(y)) = \frac{(y^2 - y) - (x^2 - x)}{y - x} = x + y - 1.
    \end{equation*}

    \noindent With this, we calculate, for all~$j \in \{0,\dots,N/2-1\}$,
    \begin{align*}
      \operatorname{slope}(\vec{v_{j,0}};\vec{v_{j,1}}) &= \frac{2Mj + 2M(j+1)-1}{MN-1} - 1\\
      &= \frac{M(2j+1) - 1 + M(2j+1)}{MN-1} - 1\\
      &= \operatorname{slope}(\vec{w_{j,0}};\vec{w_{j,1}}).
    \end{align*}
    Analogously, for all~$j \in \{0,\dots,N/2-2\}$,
    \begin{align*}
      \operatorname{slope}(\vec{v_{j,1}};\vec{v_{j+1,0}}) &= 
      \frac{2M(j+1)-1 + 2M(j+1)}{MN-1} - 1\\
      &= \frac{M(2j+1) + M(2j+3) - 1}{MN-1} - 1\\
      &= \operatorname{slope}(\vec{w_{j,1}};\vec{w_{j+1,0}}).
    \end{align*}
    Finally, we have
    \begin{align*}
      \operatorname{slope}(\vec{v_{N/2-1,1}};\vec{v_{0,0}}) &= 
      \frac{(NM - 1) + 0}{MN-1} - 1\\
      &= \frac{M(N-1) + (M - 1)}{MN-1} - 1\\
      &= \operatorname{slope}(\vec{w_{N/2-1,1}};\vec{w_{0,0}}).
    \end{align*}
\end{proof}

For our construction, we fix~$n,d \in \N$ such that~$d$ and~$n/(2d) \geq 2$ are even integers.
We recursively define polyhedra~$\mathcal{Q}_i \subseteq \R^i$ for even values of~$i$, together with a projection~$(\varphi_i,\varphi'_i)\colon \R^i \to \R^2$, as follows.
\begin{align*}
    \mathcal{Q}_2 &\coloneqq \conv(\mathcal{V}_{2, \frac{n}{2d}} \cup \mathcal{W}_{2, \frac{n}{2d}}), \\
    \varphi_2( \vec x) &\coloneqq x_1, \\
    \varphi_2'( \vec x) &\coloneqq x_2.
\end{align*}
Let~$N \coloneqq n/d$ and~$M_i \coloneqq (n/d)^{i/2}$ for all~$i\in\{2,4,\dots,d\}$.
For all~$i\in\{2,4,\dots,d-2\}$, we define
\begin{align}
    \mathcal{Q}_{i+2} &\coloneqq (\mathcal{Q}_i, \varphi_i) \bowtie (\mathcal{V}_{M_i,N},\, \mathcal{W}_{M_i,N}),\label{def:Q_deformed_prod} \\
    \varphi_{i+2}(\vec x) &\coloneqq x_{i+1},\label{def:Q_proj1} \\
    \varphi_{i+2}'(\vec x) &\coloneqq x_{i+2} + \left(\frac{M_i-1}{M_{i+2}-1}\right)^2\varphi_i'((x_1,\dots,x_i)^{\top}).\label{def:Q_proj2}
\end{align}

\begin{figure}
    \centering
    \begin{tikzpicture}
  \begin{axis}[
    width=10cm,
    height=5cm,
    xlabel={$\varphi$},
    ylabel={$\varphi'$},
    xlabel style={at={(axis description cs:0.5,0.1)}, anchor=north},
    ylabel style={at={(axis description cs:0.1,0.5)}, anchor=south},
    xtick={0,1},
    xticklabels={$0$,$1$},
    ytick={-0.25,0},
    yticklabels={$-0.25$,$0$},
    grid=both,
    axis on top,
    clip=false
  ]

    \definecolor{lightpurple}{RGB}{170,0,255}

    \addplot [color=lightpurple, mark=*, thick]
      coordinates {
        (0/15,{0/15^2 - 0/15})
        (1/15,{(1/15)^2 - (1/15)})
        (2/15,{(2/15)^2 - (2/15)})
        (3/15,{(3/15)^2 - (3/15)})
        (4/15,{(4/15)^2 - (4/15)})
        (5/15,{(5/15)^2 - (5/15)})
        (6/15,{(6/15)^2 - (6/15)})
        (7/15,{(7/15)^2 - (7/15)})
        (8/15,{(8/15)^2 - (8/15)})
        (9/15,{(9/15)^2 - (9/15)})
        (10/15,{(10/15)^2 - (10/15)})
        (11/15,{(11/15)^2 - (11/15)})
        (12/15,{(12/15)^2 - (12/15)})
        (13/15,{(13/15)^2 - (13/15)})
        (14/15,{(14/15)^2 - (14/15)})
        (15/15,{(15/15)^2 - (15/15)})
        (0/15,{0/15^2 - 0/15}) 
      };

  \end{axis}
\end{tikzpicture}
    \caption{Projection of~$\mathcal{Q}_i$ to the plane $(\varphi_i(\vec x),\, \varphi_i'(\vec x))$ for $n = 16$, $d = 4$, $i = 4$.\label{fig:projection_of_Q}}
\end{figure}

We prove that~$\mathcal{Q}_i$ has an exponential number of vertices with respect to its description size (i.e., its number of facets), and that it projects to the parabola, as desired (see Figure~\ref{fig:projection_of_Q}).

\begin{lemma}\label{lem:Q}
For every $i\in\{ 2,4,\dots,d\}$, we have
    \begin{enumerate}
        \item $\varphi_i(\mathcal{Q}_i) = [0,1]$,\label{lem:Q_item1}
        \item $\mathcal{Q}_i\subseteq\R^i$ has $\frac{in}{2d}$ facets and $(n/d)^{i/2}=M_i$ vertices,\label{lem:Q_item2}
        \item the linear map
$L_i\coloneqq(\varphi_i,\varphi_i')\colon\R^i\to\R^2$
satisfies\label{lem:Q_item3} 
        \begin{equation}
        L_i(\mathcal Q_i)= 
        \conv\left\{ {\scriptstyle\begin{pmatrix}\frac{t}{M_i-1}\\[4pt] \left(\frac{t}{M_i-1}\right)^2-\frac{t}{M_i-1}\end{pmatrix}} \colon t\in\{0,1,\dots,M_i-1\}\right\}.
        \label{eq:projection}\end{equation}
    \end{enumerate}
\end{lemma}
\begin{proof}
    We prove \ref{lem:Q_item1} via induction on~$i$. For $i=2$, we have
    \begin{equation*}
    \begin{aligned}
        \varphi_2( \conv(\mathcal{V}_{2, \frac{n}{2d}} \cup \mathcal{W}_{2, \frac{n}{2d}})) = &~\varphi_2(\conv\{\vec{v_{0,0}}, \vec{v_{n/(4d)-1,1}}\}) \\ = &~\conv\{\varphi_2(\vec{v_{0,0}}), \varphi_2(\vec{v_{n/(4d)-1,1}})\} = [0,1].
    \end{aligned}
    \end{equation*} 
    Now, for~$i+2$, by induction, the assumption of Theorem~\ref{lem:AZ} holds for the deformed product in~\eqref{def:Q_deformed_prod}.
    Since~$\mathcal{V}_{M_i,N}, \mathcal{W}_{M_i,N} \subseteq \R^2$, Theorem~\ref{lem:AZ} \ref{lem:AZ_item1} and~\eqref{def:Q_proj1} yield that every value in~$\varphi_{i+2}(\mathcal{Q}_{i+2})$ is given by the convex combination of first coordinates of all vectors of the form~$\vec{v_{j,\ell}} + \varphi_i(\vec{p})(\vec{w_{j,\ell}} - \vec{v_{j,\ell}})$ with~$\vec{p} \in \mathcal{Q}_i$.
    All vectors of this form have first coordinates between 0 and 1 since the first coordinates of all~$\vec{v_{j,\ell}}$ and~$\vec{w_{j,\ell}}$ lie in this interval and, by induction,~$\varphi_i(\mathcal{Q}_i)=[0,1]$.
    By induction, there is a choice of~$\vec{p}$ with~$\varphi_i(\vec{p}) = 0$, which means that we only need to identify vectors of the form~$\vec{v_{j,\ell}}$ with first coordinates~$0$ and~$1$, respectively.
    The vectors~$\vec{v_{0,0}}$ and~$\vec{v_{N/2-1,1}}$ are a suitable choice.

    For \ref{lem:Q_item2}, note that, by Theorem~\ref{lem:AZ} and \ref{lem:Q_item1}, $\mathcal{V}_{M,N}, \mathcal{W}_{M,N} \subseteq \R^2$ inductively implies~$\mathcal{Q}_i \subseteq\R^i$.
    Since~$\mathcal{V}_{M,N}$ and~$\mathcal{W}_{M,N}$ are polygons with~$N$ vertices, they can each be expressed by systems of~$N$ inequalities.
    Theorem~\ref{lem:AZ} \ref{lem:AZ_item2} yields that each deformed product in~\eqref{def:Q_deformed_prod} adds~$N$ inequalities and thus facets to the description of~$\mathcal{Q}_i$, yielding~$iN/2 = \frac{in}{2d}$ facets overall.
    Last, Theorem~\ref{lem:AZ} \ref{lem:AZ_item3} implies that each deformed product multiplies the number of vertices by~$N$, for a total number of~$N^{i/2}=(n/d)^{i/2}$ vertices of~$\mathcal{Q}_i$.

    For \ref{lem:Q_item3}, first observe that the right-hand side of~\eqref{eq:projection}
    has the same number of vertices as~$\mathcal{Q}_i$.
    Therefore, it suffices to prove the following two statements
    for every~$i\in\{2,4,\dots,d\}$:
    \begin{itemize}[leftmargin=2.5em]
        \item for every~$t \in\{ 0,1,\dots,M_i-1\}$, there exists a vertex~$\vec q \in \mathcal{Q}_i$ with~$\varphi_i(\vec q) = \frac{t}{M_i-1}$,
    \item for every vertex~$\vec q \in \mathcal{Q}_i$, we have~$\varphi_i'(\vec q) = \varphi_i(\vec q)^2 - \varphi_i(\vec q)$.
    \end{itemize}
    We prove these by induction.
    The statements are true for~$i=2$ by inspection.
    Now, assuming they are true for~$i$, we prove them for~$i+2$.
    Let~$t \in \{0,1,\dots,M_{i+2}-1\}$ and denote by~$k \in \{0,1,\ldots,N-1\}$ the unique number such that~$kM_i \leq t \leq (k+1)M_i-1$.
    Observe that such a~$k \leq N-1$ exists as~$NM_i=M_{i+2}$.

    We write~$k=2j+\ell$, where~$j\in\Z_{\geq 0}$ and~$\ell\in\{0,1\}$ such that
    \[
    (2j+\ell)M_i \leq t \leq (2j+\ell+1) M_i - 1.
    \]
    Further, we define
    \begin{equation}\label{eq:def s}
    s \coloneqq (1-\ell)(t - 2jM_i)+\ell((2j+2)M_i-1-t) \in \Z
    \end{equation}
    and note that~$0 \leq s \leq M_i-1$.
    Hence, by the induction hypothesis, there exists
    a vertex~$\vec q \in \mathcal{Q}_i$ such that~$\varphi_i(\vec q) = \frac{s}{M_i-1}$.
    By Theorem~\ref{lem:AZ}, we know that the point~$\vec p \coloneqq \spmatrix{\vec q \\ \vec{v_{j,\ell}} + \varphi_i(\vec q)(\vec{w_{j,\ell}}-\vec{v_{j,\ell}})}$
    is a vertex of~$\mathcal{Q}_{i+2}$ and it satisfies
    \begin{align*}
    \varphi_{i+2}(\vec p) &=\varphi_2(\vec{v_{j,\ell}}) + \varphi_i(\vec q)(\varphi_2(\vec{w_{j,\ell}})-\varphi_2(\vec{v_{j,\ell}})) \\ 
    &= \frac{2M_i(j+\ell)-\ell}{M_{i+2}-1}+\frac{s}{M_i-1}\cdot\frac{(M_i-1)(1-2\ell)}{M_{i+2}-1} \\
    &= \frac{t}{M_{i+2}-1},      
    \end{align*}
    which concludes the induction for the first statement.

    Note that~$\mathcal{Q}_{i+2}$ has~$M_{i+2}$ vertices by~\ref{lem:Q_item2}.
    Thus, the first inductive statement yields a bijection between the integers from~$0$ to~$M_{i+2}-1$ and the vertices of~$\mathcal{Q}_{i+2}$.
    Since in the argument above~$t \in \{0,1,\dots,M_{i+2}-1\}$ was chosen arbitrarily and~$\vec p$ is the vertex of~$\mathcal{Q}_{i+2}$ corresponding to~$t$, for the second statement it  suffices to prove
    \begin{equation}\label{eq:appendix}
    \varphi_{i+2}'(\vec p) = \varphi_{i+2}(\vec p)^2 - \varphi_{i+2}(\vec p).        
    \end{equation}

    For this technical part of the proof, we first carry out some preparatory computations, which we then combine to obtain the desired identity.
    In the following, note that~$\ell\in\{0,1\}$ implies~$\ell=\ell^2$.
    
    First, we have
    \begin{align}
    \varphi_{2}'(\vec{v_{j,\ell}}) &= ~\left(\varphi_{2}(\vec{v_{j,\ell}})\right)^2-\varphi_{2}(\vec{v_{j,\ell}})\notag\\
    &=~\left(\frac{2M_i(j+\ell)-\ell}{M_{i+2}-1}\right)^2 
        -\frac{(2M_i(j+\ell)-\ell)(M_{i+2}-1)}{(M_{i+2}-1)^2}\notag\\
        &=~ \frac{4j^2M_i^2+8j\ell M_i^2-4\ell j M_i+4\ell M_i^2-4 \ell M_i}{(M_{i+2}-1)^2}\notag\\
        &\hspace{40pt} {}+\frac{-2jM_iM_{i+2}-2\ell M_iM_{i+2}+\ell M_{i+2}+2jM_i+2\ell M_i}{(M_{i+2}-1)^2}.\label{eq:LHS1}
    \end{align}
    
    Second, due to
    \begin{align*}
    \varphi_{2}'(\vec{w_{j,\ell}})-\varphi_{2}'(\vec{v_{j,\ell}})~ &=~ \left(\varphi_{2}(\vec{w_{j,\ell}})\right)^2-\varphi_{2}(\vec{w_{j,\ell}})-\left(\varphi_{2}(\vec{v_{j,\ell}})\right)^2+\varphi_{2}(\vec{v_{j,\ell}})\\
    &=~ (\varphi_{2}(\vec{w_{j,\ell}})-\varphi_{2}(\vec{v_{j,\ell}}))(\varphi_{2}(\vec{w_{j,\ell}})+\varphi_{2}(\vec{v_{j,\ell}})-1) \\
    &=~\frac{(M_i-1)(1-2\ell)}{M_{i+2}-1}\left(\frac{4jM_i+2\ell M_i+M_i-1-(M_{i+2}-1)}{M_{i+2}-1}\right)
    \end{align*}
    and~$\varphi_i(\vec q) = \frac{s}{M_i-1}$ we have
    \begin{align}
    &\phantom{=} ~~~\varphi_i(\vec q)(\varphi_{2}'(\vec{w_{j,\ell}})-\varphi_{2}'(\vec{v_{j,\ell}})) \notag\\
    &=~\frac{s(1-2\ell)}{M_{i+2}-1}\left(\frac{4jM_i+2\ell M_i+M_i-1-(M_{i+2}-1)}{M_{i+2}-1}\right)\notag\\
    &=~\frac{4jsM_i+2\ell sM_i+sM_i-sM_{i+2}-8s\ell jM_i-4s\ell M_i-2s\ell M_i+2s\ell M_{i+2}}{(M_{i+2}-1)^2}.\label{eq:LHS2}
    \end{align}
    
    Third, by induction,
    \[
    \varphi_i'(\vec q) = \varphi_i(\vec q)^2 - \varphi_i(\vec q) =\left(\frac{s}{M_i-1}\right)^2 - \frac{s}{M_i-1},
    \]
    and we thus have
    \begin{equation}\label{eq:LHS3}
    \left(\frac{M_i-1}{M_{i+2}-1}\right)^2\varphi_i'(\vec q) = \frac{s^2-sM_i+s}{(M_{i+2}-1)^2}.    
    \end{equation}
    
    Fourth, \eqref{eq:def s} yields, for both values of~$\ell \in\{0,1\}$, 
    \[
    t=(2\ell-1)(4j\ell M_i-2jM_i+2\ell M_i-\ell -s).
    \]
    With~$(2\ell-1)^2=1$, since~$\ell\in\{0,1\}$, we get
    \begin{align}
    t^2 &=~ 8j\ell M_i^2-4j\ell M_i-8j\ell sM_i+4j^2M_i^2+4jsM_i\notag\\
    &\hspace{80pt} {}+ 4M_i^2\ell-4\ell M_i-4\ell s M_i+\ell+2\ell s+s^2.\label{eq:RHS1}
    \end{align}
    
    Fifth, we have
    \begin{align}
    -t(M_{i+2}-1) &=~ -2\ell M_iM_{i+2}+\ell M_{i+2}+2\ell sM_{i+2}-2jM_iM_{i+2} \notag \\
    &\hspace{70pt} {}-sM_{i+2}+2\ell M_i-\ell-2\ell s+2j M_i+s. \label{eq:RHS2}   
    \end{align}
    
    By combining~\eqref{eq:LHS1}--\eqref{eq:RHS2}, we finally obtain
    \begin{align*}
        \varphi_{i+2}'(\vec p) ~ &=~ \varphi_{2}'(\vec{v_{j,\ell}})+\varphi_i(\vec q)(\varphi_{2}'(\vec{w_{j,\ell}})-\varphi_{2}'(\vec{v_{j,\ell}}))
        + \left(\frac{M_i-1}{M_{i+2}-1}\right)^2\varphi_i'(\vec q)\\
        &=~ \frac{t^2-t(M_{i+2}-1)}{(M_{i+2}-1)^2}\\
        &=~ \varphi_{i+2}(\vec p)^2 - \varphi_{i+2}(\vec p),
    \end{align*}
    proving equation~\eqref{eq:appendix}.
    This concludes the proof.
\end{proof}

By Lemma~\ref{lem:Q}\ref{lem:Q_item3}, the linear map~$L_d\coloneqq(\varphi_d,\varphi_d')\colon\R^d\to\R^2$ has rank two.
After an invertible linear change of variables,~$L_d$ is the standard coordinate projection onto the first two coordinates, which is orthogonal.
Since linear isomorphisms preserve facets and vertices, we can extract the main result of this section by setting~$i=d$ in Lemma~\ref{lem:Q}.
This proves the following theorem (restated Theorem~\ref{thm:main2}).

\begin{theorem}\label{thm:extended parabola}
    For every~$n,\,\! d \in \mathbb N$ such that $d$ and~$n/(2d) \geq 2$ are even integers, there exists a polytope~$\mathcal{P} \subseteq \R^d$ with~$n/2$ facets and~$\smash{M\coloneqq\left(n/d\right)^{d/2}}$ vertices such that the orthogonal projection~$\Pi\colon\R^d\to\R^2$,~$\Pi(\vec x)=(x_1,x_2)^\top$, satisfies
    \[
    \Pi(\mathcal{P})=\conv\left\{\begin{pmatrix} \frac{t}{M-1}\\ \left(\frac{t}{M-1}\right)^2 - \frac{t}{M-1}
    \end{pmatrix}\colon t\in\{0,1,\ldots,M-1\}\right\}.
    \]
\end{theorem}

We emphasize that none of the vertices of the extended formulation project to the interior of its shadow.

\begin{corollary}\label{cor:vertex bijection}
    Consider the polytope~$\mathcal{P}$ and the projection~$\Pi$ of Theorem~\ref{thm:extended parabola}.
    There is a bijection between the vertices of~$\mathcal{P}$ and the vertices of~$\Pi(\mathcal{P})$.
\end{corollary}
\begin{proof}    
Every face of the shadow~$\Pi(\mathcal{P})$ is the projection of a face of the polytope~$\mathcal{P}$.
In particular, every vertex of the shadow is the projection of a vertex of~$\mathcal{P}$, which defines a bijection as both polytopes have~$M$ vertices.
\end{proof}

\subsection{Active-set on the extended parabola\label{sec:lower_bound}}

We now construct the objective function that prevents shortcuts along chords of the projected polygon, as outlined at the beginning of this section.
Together with the extended formulation from Section~\ref{sec:extended_form}, this yields an exponential worst-case lower bound for the active-set method in convex maximization with objectives of the form~$(x_1)^{1+\eps}+\vec c^\top\vec x$, for every fixed~$0<\eps\leq1$.
The case~$\eps=1$ gives the corresponding lower bound for convex quadratic maximization.

Let~$\mathcal P\subseteq\R^d$ be the polytope of Theorem~\ref{thm:extended parabola} with~$n=4d$, where~$d\geq2$ is even.
Then~$\mathcal P$ has~$2d$ facets and~\mbox{$M=2^d$} vertices, is combinatorially equivalent to the~$d$-dimensional hypercube, and its coordinate projection~$\Pi(\vec x)=(x_1,x_2)^\top$ maps its vertices
bijectively onto the points
\begin{equation*}
    \vec{x^{(t)}}\coloneqq
    \begin{pmatrix}
        \frac{t}{M-1}\\
        \left(\frac{t}{M-1}\right)^2-\frac{t}{M-1}
    \end{pmatrix},
    \qquad t\in\{0,\ldots,M-1\}.
\end{equation*}

The geometry of the projected parabola~$x_2=(x_1)^2-x_1$ suggests how to choose the objective.
We start with the intuition for the quadratic case~$\eps=1$.
The quadratic function~$(x_1)^2-x_1-x_2$ vanishes on the parabola and decreases initially along every line segment joining two distinct projected vertices, since the segment enters the region above the parabola.
Adding a small positive linear term in~$x_1$ rewards movement towards larger first coordinates.
With a suitable perturbation, this reward outweighs the initial decrease when moving towards the next vertex, but not when moving towards any more distant vertex.
Movement towards vertices with smaller first coordinates remains initially decreasing.
This forces \textsc{ActiveSet} to follow consecutive boundary edges of the projection (see Figure~\ref{fig:QP}).

\begin{figure}
\centering
  \begin{tikzpicture}[scale=1, font=\small]
\begin{axis}[
  view={170}{20},
  xlabel={$x_1$}, ylabel={$x_2$}, zlabel={objective},
  xtick={0,1}, ytick={-0.24,0}, ztick={-0.2,0,0.1},
  xlabel style={at={(axis description cs:0.6,0.0)}, anchor=north},
  ylabel style={at={(axis description cs:0.08,0.08)}, anchor=east},
  zlabel style={at={(axis description cs:-0.1,0.7)}, anchor=east},
  tick align=outside,
  x dir=reverse,
  grid=none,
]

\addplot3[lightpurple, thick, mark=none, samples=20, samples y=1, domain=0:1]
({x}, {x^2 - x}, {x^2 - (1 - 3/(2*16-2))*x - (x^2 - x)});

\addplot3[only marks, mark size=2pt, mark=*, mark options={fill=lightpurple,draw=lightpurple}, draw=none]
coordinates {
(0.00000,  0.00000,  0.00000)
(0.06667, -0.06222,  0.0066)
(0.13333, -0.11556,  0.01333)
(0.20000, -0.16000,  0.02)
(0.26667, -0.19556,  0.02667)
(0.33333, -0.22222,  0.03333)
(0.40000, -0.24000,  0.04)
(0.46667, -0.24889,  0.04667)
(0.53333, -0.24889,  0.05333)
(0.60000, -0.24000,  0.06)
(0.66667, -0.22222,  0.06667)
(0.73333, -0.19556,  0.07333)
(0.80000, -0.16000,  0.08)
(0.86667, -0.11556,  0.08667)
(0.93333, -0.06222,  0.09333)
(1.00000,  0.00000,  0.1)
};

\addplot3[myorange, thick, mark=none, samples=20, samples y=1, domain=0:1]
({x}, {0}, {x^2 - (1 - 3/(2*16-2))*x});

\addplot3[myorange, thick, mark=none, samples=20, samples y=1, domain=0.06667:0.93333]
({x}, {-0.06222}, {x^2 - (1 - 3/(2*16-2))*x + 0.06222});

\addplot3[myorange, thick, mark=none, samples=20, samples y=1, domain=0.13333:0.8666]
({x}, {-0.1155}, {x^2 - (1 - 3/(2*16-2))*x + 0.1155});

\addplot3[myorange, thick, mark=none, samples=20, samples y=1, domain=0.2:0.8]
({x}, {-0.16}, {x^2 - (1 - 3/(2*16-2))*x + 0.16});

\addplot3[myorange, thick, mark=none, samples=20, samples y=1, domain=0.2666:0.7333]
({x}, {-0.1955}, {x^2 - (1 - 3/(2*16-2))*x + 0.1955});

\addplot3[myorange, thick, mark=none, samples=20, samples y=1, domain=0.3333:0.6666]
({x}, {-0.2222}, {x^2 - (1 - 3/(2*16-2))*x + 0.2222});

\addplot3[myorange, thick, mark=none, samples=20, samples y=1, domain=0.4:0.6]
({x}, {-0.24}, {x^2 - (1 - 3/(2*16-2))*x + 0.24});

\end{axis}
\end{tikzpicture}
\caption{%
Objective function values of~$(x_1)^2-\bigl(1-\frac{3}{2(M-1)}\bigr)x_1-x_2$ along
vertices and edges of the projection~$\Pi(\mathcal{P})$ for~$n=M=16$ and~$d=4$ (the case~$\eps=1$ without the compression~$T$).
Purple edges are monotonically increasing in objective value, while the orange edge and chords have negative directional derivatives from either endpoint towards the other.\label{fig:QP}}
\end{figure}

For~$0<\eps<1$, the term~$(x_1)^{1+\eps}-x_1$ no longer matches the projected parabola exactly.
We recover the same geometric argument locally near~$x_1=1$ by compressing the projection onto a short portion of the parabola and choosing the linear coefficients to cancel the slope and curvature there, apart from a small positive perturbation.
To this end, set
\begin{equation*}
    \delta\coloneqq\frac{1}{16(M-1)},
    \qquad
    h\coloneqq\frac{\delta}{M-1}=\frac{1}{16(M-1)^2}.
\end{equation*}
Here,~$\delta$ is the length of the interval containing the first coordinates of the transformed vertices, and~$h$ is their spacing.
Define the affine map~$T\colon\R^d\to\R^d$ by
\begin{equation*}
    T(\vec z)\coloneqq
    \bigl(1-\delta+\delta z_1,\,
    \delta^2z_2+\delta(1-\delta)(z_1-1),\,
    z_3,\ldots,z_d\bigr)^\top,
\end{equation*}
and let~$\mathcal P_\delta\coloneqq T(\mathcal P)$.
The linear part of~$T$ has determinant~$\delta^3>0$, so~$T$ preserves the dimension, the number of facets, and the face structure of~$\mathcal P$.
Moreover, for every~$u\in [0,1]$, we have
\[
    \delta^2(u^2-u)+\delta(1-\delta)(u-1) = (1-\delta+\delta u)^2 - (1-\delta+\delta u).
\]
Consequently, the vertices of~$\mathcal P_\delta$ project bijectively onto
\begin{equation}\label{eq:vertices epsilon}
    \vec{y^{(t)}}\coloneqq
    \begin{pmatrix}s_t\\s_t^2-s_t\end{pmatrix},
    \qquad
    s_t\coloneqq1-\delta+th,
    \qquad t\in\{0,\ldots,M-1\}.
\end{equation}
Thus, all projected vertices remain on the original parabola, with first coordinates in~$[1-\delta,1]$.
This transformation requires no additional variables or constraints.

We seek an objective of the form~$f_\eps(x_1,x_2)=(x_1)^{1+\eps}+c_1x_1+c_2x_2$.
Since~$x_1>0$ on~$\mathcal P_\delta$, replacing~$x_1^{1+\eps}$ by~$|x_1|^{1+\eps}$ gives a convex, continuously differentiable extension to~$\R^d$ without changing its values or gradients on the feasible region.

Along the parabola, the objective becomes~$s^{1+\eps}+c_1s+c_2(s^2-s)$.
At~$s=1$, the term~$s^{1+\eps}$ has slope~$1+\eps$ and curvature~$(1+\eps)\eps$, whereas~$s^2-s$ has slope~$1$ and
curvature~$2$.
We therefore choose~$c_2$ to cancel the curvature and~$c_1$ to cancel the resulting slope, apart from a small positive perturbation.
With these coefficient choices, the objective along the parabola is close to an increasing affine function of~$x_1$ throughout the short interval~$[1-\delta,1]$ containing the transformed projected vertices.
Straight-line movement between vertices, however, follows segments above the parabola, which the negative coefficient~$c_2$ penalizes more strongly as the vertices become farther apart.
To make movement towards the next vertex improving while excluding shortcuts, we choose the positive slope perturbation to be~$(3(1+\eps)\eps h)/4$.
This gives
\begin{equation}\label{eq:coefficients epsilon}
    c_1\coloneqq-(1+\eps)+\frac{(1+\eps)\eps}{2}
        +\frac{3(1+\eps)\eps h}{4},
    \qquad
    c_2\coloneqq-\frac{(1+\eps)\eps}{2},
\end{equation}
and hence the convex maximization problem
\begin{equation}
\begin{aligned}
    \max f_\eps(x_1,x_2)\coloneqq (x_1)^{1+\eps}+c_1x_1+c_2x_2,
    \quad\text{subject to }\vec x\in\mathcal P_\delta.
\end{aligned}
\tag{\ensuremath{\mathrm{LP}_{\varepsilon,\mathcal P_\delta}}}
\label{LPeps}
\end{equation}
The following lemma shows that, at every projected vertex, movement towards the next vertex is improving, while movement towards any other vertex initially decreases the objective.

\begin{lemma}\label{lem:improving directions epsilon}
    For every~$t,k\in\Z$ with~$0\leq t,t+k\leq M-1$, we have
    \[
    \nabla f_\eps(\vec{y^{(t)}})^\top
        (\vec{y^{(t+k)}}-\vec{y^{(t)}})
    \begin{cases}
        >0, & \text{if }k=1,\\
        <0, & \text{if }k\notin\{0,1\}.
    \end{cases}
    \]
\end{lemma}
\begin{proof}
Fix~$t,k\in\Z$ satisfying the conditions in the statement.
By~\eqref{eq:vertices epsilon}, we have
\begin{equation*}
    \vec{y^{(t+k)}}-\vec{y^{(t)}}
    =kh\begin{pmatrix}1\\2s_t+kh-1\end{pmatrix}.
\end{equation*}
Further, the gradient of the objective at~$\vec x\in\R^2$ with~$x_1>0$ is given by
\begin{equation*}
    \nabla f_\eps(\vec x)
    =\begin{pmatrix}(1+\eps)x_1^\eps+c_1\\c_2\end{pmatrix}.
\end{equation*}
Combining these identities yields
\begin{align*}
    \nabla f_\eps(\vec{y^{(t)}})^\top
        (\vec{y^{(t+k)}}-\vec{y^{(t)}})
    &=kh\bigl((1+\eps)s_t^\eps+c_1+c_2(2s_t+kh-1)\bigr)\\
    &=kh\bigl(g(s_t)+c_2kh\bigr),
\end{align*}
where we write~$g(s)\coloneqq(1+\eps)s^\eps+c_1+c_2(2s-1)$ for~$s\geq0$.
Since~$c_2<0$, it suffices to prove
\begin{equation}\label{eq:inequalities epsilon}
    -c_2h<g(s)<-2c_2h
    \qquad\text{for all }s\in[1-\delta,1].
\end{equation}
Indeed, these inequalities make the expression in parentheses positive for~$k=1$, negative for~$k\geq2$, and positive for~$k<0$; in the last case, the prefactor~$kh$ is negative.

We prove~\eqref{eq:inequalities epsilon} by bounding~$g$ at the endpoints of the interval.
For~$s\in(0,1]$, we have
\[
    g'(s)=(1+\eps)\eps(s^{\eps-1}-1)\geq0,
\]
so~$g$ is monotonically increasing on~$(0,1]$.
At the right endpoint,~\eqref{eq:coefficients epsilon} gives
\[
    g(1)=(1+\eps)+c_1+c_2
    =-\frac32c_2h<-2c_2h.
\]
At the left endpoint, we obtain
\[
    g(1-\delta)
    =g(1)+(1+\eps)\bigl((1-\delta)^\eps-1+\eps\delta\bigr).
\]
Recall that~$\ln(1-u)\geq-u/(1-u)$ for~$u\in[0,1)$ and~$e^v\geq1+v$ for~$v\in\R$.
These inequalities imply
\[
    (1-u)^\eps=e^{\eps\ln(1-u)}
    \geq1-\frac{\eps u}{1-u},
    \qquad
    (1-u)^\eps-1+\eps u
    \geq-\frac{\eps u^2}{1-u}.
\]
Since~$\delta\leq1/4$ and~$\delta^2=h/16$, we conclude that
\begin{align*}
    g(1-\delta)
    &\geq-\frac32c_2h
        -\frac{(1+\eps)\eps\delta^2}{1-\delta}\\
    &\geq-\frac32c_2h
        -\frac43(1+\eps)\eps\frac{h}{16}\\
    &=-\frac32c_2h+\frac16c_2h
     =-\frac43c_2h>-c_2h.
\end{align*}
Together with monotonicity, this proves~\eqref{eq:inequalities epsilon} and hence the lemma.
\end{proof}

We conclude that there is a unique improving edge at each vertex of the higher-dimensional polytope, and that these edges form a path that visits every vertex.

\begin{lemma}\label{lem:long hirsch path epsilon}
    At every vertex~$\vec x$ of~$\mathcal P_\delta$ with~$\Pi(\vec x)\neq\vec{y^{(M-1)}}$, there is a unique improving edge with respect to the objective in~\eqref{LPeps}.
    These improving edges form a path~$\mathcal H$ in the graph of~$\mathcal P_\delta$ that visits all~$M$ vertices and projects to
    \[
        \Pi(\mathcal H)
        =\bigl(\vec{y^{(0)}},\vec{y^{(1)}},\ldots,\vec{y^{(M-1)}}\bigr).
    \]
\end{lemma}
\begin{proof}
Every consecutive pair~$\vec{y^{(t)}},\vec{y^{(t+1)}}$ forms an edge of the projected polygon.
The preimage of this edge in~$\mathcal P_\delta$ is a face with exactly two vertices, since the projection bijects the vertex sets (see Corollary~\ref{cor:vertex bijection}).
It is therefore an edge of~$\mathcal P_\delta$.
Conversely, every edge of~$\mathcal P_\delta$ projects to a segment between two distinct projected vertices.
Since the objective depends only on the first two coordinates, Lemma~\ref{lem:improving directions epsilon} implies that the edges between consecutive vertices are precisely the improving edges, directed in increasing order of~$t$.
\end{proof}

Thus the improving edge chosen in each step of \textsc{ActiveSet} for~\eqref{LPeps} is uniquely determined. 
Therefore, the behavior of \textsc{ActiveSet} does not
depend on the choice of a pivot rule.
We conclude our main result: the active-set method takes an exponential number of iterations for our construction, irrespective of the pivot rule (proving Theorem~\ref{thm:main epsilon}).

\begin{theorem}\label{thm:main_details epsilon}
    For every~$0<\eps\leq1$, \textup{\textsc{ActiveSet}}, when started at the unique vertex~$\vec{\bar x}\in\mathcal P_\delta$ such that~$\Pi(\vec{\bar x})=\vec{y^{(0)}}$, takes~$2^{m/2}-1$ iterations to find a KKT point of~\eqref{LPeps}, irrespective of the pivot rule, where~$m$ denotes the number of facets of~$\mathcal P_\delta\subseteq\R^{m/2}$.
\end{theorem}
\begin{proof}
The polytope~$\mathcal P_\delta$ is simple, and Lemma~\ref{lem:long hirsch path epsilon} gives a unique improving edge at each vertex of the path except the last.
Thus, the choice of an improving edge is independent of the pivot rule.
Since~$f_\eps$ is convex, its directional derivative along an improving edge~$(\vec x,\vec y)$ is nondecreasing; in particular we have
    \[
    \nabla f_{\eps}(\vec x)^\top(\vec y-\vec x)>0 \implies \nabla f_{\eps}(\vec x+\lambda(\vec y-\vec x))^\top(\vec y-\vec x)>0 \quad\text{for all }\lambda\in[0,1],
    \]
so the method traverses each chosen edge to its other endpoint.
At the last vertex, Lemma~\ref{lem:improving directions epsilon} implies that the directional derivative towards every other vertex is negative.
Hence, there is no feasible improving direction, and the last vertex is a KKT point.
The method therefore visits all~$M=2^d$ vertices in order and takes exactly~$M-1=2^{m/2}-1$ iterations, since~$m=2d$.
\end{proof}

\section{Future research}

We may also relate our result to disproving the polynomial monotone Hirsch conjecture~\cite{dantzig1963linear}, by extending the definition of the monotone diameter of a polytope to convex objective functions.
Lemma~\ref{lem:long hirsch path epsilon} then yields the following.

\begin{corollary}
   For every fixed~$\eps\in(0,1]$, the largest possible monotone diameter with respect to a convex objective of the form~$(x_1)^{1+\eps}+\vec c^{\top}\vec x$ is exponential.
\end{corollary}

While a general bound for linear objectives, i.e., for the simplex method, would be a major breakthrough, reaching quadratic concave objectives would already be very interesting:
much like in the linear setting, only weakly polynomial algorithms are known, namely ellipsoid and interior-point methods~\cite{kozlov1979polynomial,ye1989extension}, and the active-set method is a promising candidate for a strongly polynomial algorithm, due to its combinatorial nature.

\begin{question}
Is there a concave, quadratic polynomial over some polytope for which the active-set method takes super-polynomially many iterations for all pivot rules?
\end{question}

\section*{Declarations}
The lower bound in convex quadratic maximization and the extended formulation, previously published in the conference version~\cite{bach2026unconditional}, were developed without AI assistance.
For the extension to exponents~$1+\eps$, the authors conceived the generalization and the approach based on local quadratic approximation near~$x_1=1$.
Within this approach, OpenAI's GPT-6 Astra, accessed through Codex, assisted in determining explicit values of the coefficients~$c_1,c_2$ and an explicit affine transformation~$T$.
The authors independently verified the resulting calculations and all mathematical arguments, and substantially restructured and rewrote the proof and its exposition.
The tool was also used for proofreading and language refinement.
The authors assume responsibility for all content.

\newpage

\bibliographystyle{siamplain}
\bibliography{references}

@article{mahjoub2019distance,
author = {Mahjoub, A. Ridha and Poss, Michael and Simonetti, Luidi and Uchoa, Eduardo},
title = {Distance Transformation for Network Design Problems},
journal = {SIAM Journal on Optimization},
volume = {29},
number = {2},
pages = {1687-1713},
year = {2019},
doi = {10.1137/16M1108261},
}

@article{baiou2002steiner,
author = {Ba{\"i}ou, Mourad and Mahjoub, Ali Ridha},
title = {The Steiner Traveling Salesman Polytope and Related Polyhedra},
journal = {SIAM Journal on Optimization},
volume = {13},
number = {2},
pages = {498-507},
year = {2002},
doi = {10.1137/S1052623400322287},
}

@article{lovasz1991cones,
author = {Lov{\'a}sz, L. and Schrijver, A.},
title = {Cones of Matrices and Set-Functions and 0–1 Optimization},
journal = {SIAM Journal on Optimization},
volume = {1},
number = {2},
pages = {166-190},
year = {1991},
doi = {10.1137/0801013},
}

@inproceedings{bach2026unconditional,
author = {Eleon Bach and Yann Disser and Sophie Huiberts and Nils Mosis},
title = {An Unconditional Lower Bound for the Active-Set Method in Convex Quadratic Maximization},
booktitle = {Proceedings of the 2026 Annual ACM-SIAM Symposium on Discrete Algorithms (SODA)},
chapter = {},
year = {},
pages = {314-327},
doi = {10.1137/1.9781611978971.14},
}

@article{lewis2002active,
author = {Lewis, A. S.},
title = {Active Sets, Nonsmoothness, and Sensitivity},
journal = {SIAM Journal on Optimization},
volume = {13},
number = {3},
pages = {702-725},
year = {2002},
doi = {10.1137/S1052623401387623},
}

@article{scheinberg2006efficient,
  author  = {Katya Scheinberg},
  title   = {An Efficient Implementation of an Active Set Method for SVMs},
  journal = {Journal of Machine Learning Research},
  year    = {2006},
  volume  = {7},
  number  = {80},
  pages   = {2237--2257},
  url     = {http://jmlr.org/papers/v7/scheinberg06a.html}
}

@article{bomze2020active,
author = {Bomze, Immanuel M. and Rinaldi, Francesco and Zeffiro, Damiano},
title = {Active Set Complexity of the Away-Step Frank--Wolfe Algorithm},
journal = {SIAM Journal on Optimization},
volume = {30},
number = {3},
pages = {2470-2500},
year = {2020},
doi = {10.1137/19M1309419},
}

@article{pena2016neumann,
author = {Pe{\~n}a, Javier and Rodr{\'i}guez, Daniel and Soheili, Negar},
title = {On the von Neumann and Frank--Wolfe Algorithms with Away Steps},
journal = {SIAM Journal on Optimization},
volume = {26},
number = {1},
pages = {499-512},
year = {2016},
doi = {10.1137/15M1009937},
}

@article{wolfe1976finding,
author = {Wolfe, Philip},
title = {Finding the nearest point in a polytope},
journal = {Mathematical Programming},
volume = {11},
pages = {128-149},
year = {1976},
doi = {10.1007/BF01580381},
}

@article{deLoera2020minimum,
author = {De Loera, Jes{\'u}s A. and Haddock, Jamie and Rademacher, Luis},
title = {The Minimum Euclidean-Norm Point in a Convex Polytope: Wolfe's Combinatorial Algorithm is Exponential},
journal = {SIAM Journal on Computing},
volume = {49},
number = {1},
pages = {138-169},
year = {2020},
doi = {10.1137/18M1221072},
}

@article{hintermuller2002primal,
author = {Hinterm{\"u}ller, M. and Ito, K. and Kunisch, K.},
title = {The Primal-Dual Active Set Strategy as a Semismooth Newton Method},
journal = {SIAM Journal on Optimization},
volume = {13},
number = {3},
pages = {865-888},
year = {2002},
doi = {10.1137/S1052623401383558},
}

@article{curtis2016globally,
author = {Curtis, Frank E. and Han, Zheng},
title = {Globally Convergent Primal-Dual Active-Set Methods with Inexact Subproblem Solves},
journal = {SIAM Journal on Optimization},
volume = {26},
number = {4},
pages = {2261-2283},
year = {2016},
doi = {10.1137/140993314},
}

@article{hungerlander2015feasible,
author = {Hungerl{\"a}nder, P. and Rendl, F.},
title = {A Feasible Active Set Method for Strictly Convex Quadratic Problems with Simple Bounds},
journal = {SIAM Journal on Optimization},
volume = {25},
number = {3},
pages = {1633-1659},
year = {2015},
doi = {10.1137/140984439},
}

@article{hager2006new,
author = { Hager, William W. and  Zhang, Hongchao},
title = {A New Active Set Algorithm for Box Constrained Optimization},
journal = {SIAM Journal on Optimization},
volume = {17},
number = {2},
pages = {526-557},
year = {2006},
doi = {10.1137/050635225},
}

@article{more1991solution,
author = {Mor{\'e}, Jorge J. and Toraldo, Gerardo},
title = {On the Solution of Large Quadratic Programming Problems with Bound Constraints},
journal = {SIAM Journal on Optimization},
volume = {1},
number = {1},
pages = {93-113},
year = {1991},
doi = {10.1137/0801008},
}

@article{burke1994exposing,
author = {Burke, James V. and Mor{\'e}, Jorge J.},
title = {Exposing Constraints},
journal = {SIAM Journal on Optimization},
volume = {4},
number = {3},
pages = {573-595},
year = {1994},
doi = {10.1137/0804032},
}

@article{Facchinei1998accurate,
author = {Facchinei, Francisco and Fischer, Andreas and Kanzow, Christian},
title = {On the Accurate Identification of Active Constraints},
journal = {SIAM Journal on Optimization},
volume = {9},
number = {1},
pages = {14-32},
year = {1998},
doi = {10.1137/S1052623496305882},
}

@article{oberlin2006active,
author = { Oberlin, Christina and Wright, Stephen J.},
title = {Active Set Identification in Nonlinear Programming},
journal = {SIAM Journal on Optimization},
volume = {17},
number = {2},
pages = {577-605},
year = {2006},
doi = {10.1137/050626776},
}

@article{fearnley2025complexity, 
author = {Fearnley, John and Goldberg, Paul and Hollender, Alexandros and Savani, Rahul}, 
title = {The Complexity of Computing {KKT} Solutions of Quadratic Programs}, 
year = {2025}, 
publisher = {Association for Computing Machinery},
volume = {72},
number = {5}, 
doi = {10.1145/3745017}, 
journal = {Journal of the ACM}, 
articleno = {31}, 
numpages = {48},
}

@article{fearnley2022complexity, 
author = {Fearnley, John and Goldberg, Paul and Hollender, Alexandros and Savani, Rahul}, 
title = {The Complexity of Gradient Descent: {CLS} = {PPAD} $\cap$ {PLS}}, 
year = {2022}, 
publisher = {Association for Computing Machinery}, 
volume = {70}, 
number = {1}, 
doi = {10.1145/3568163}, 
journal = {Journal of the ACM}, 
articleno = {7}, 
numpages = {74},
}

@inproceedings{disser2025unconditional,
  title={An unconditional lower bound for the active-set method on the hypercube},
  author={Disser, Yann and Mosis, Nils},
  booktitle={Proceedings of the 26th Conference on Integer Programming and Combinatorial Optimization (IPCO)},
  pages={213--227},
  year={2025},
  doi={10.1007/978-3-031-93112-3\_16}
}

@article{ben2001polyhedral,
  title={On polyhedral approximations of the second-order cone},
  author={Ben-Tal, Aharon and Nemirovski, Arkadi},
  journal={Mathematics of Operations Research},
  volume={26},
  number={2},
  pages={193--205},
  year={2001},
  publisher={INFORMS},
  doi = {10.1287/moor.26.2.193.10561}
}

@article{pardalos1991quadratic,
  title={Quadratic programming with one negative eigenvalue is {NP}-hard},
  author={Pardalos, Panos M and Vavasis, Stephen A},
  journal={Journal of Global Optimization},
  volume={1},
  number={1},
  pages={15--22},
  year={1991},
doi ={10.1007/BF00120662}
}

@article{black2025exponentiallowerbounds,
  title={Exponential lower bounds for many pivot rules for the simplex method},
  author={Black, Alexander E},
  journal={Mathematical Programming},
  pages={1--22},
  year={2026},
  publisher={Springer},
  doi={10.1007/s10107-026-02325-5}
}

@article{disser2018simplex,
  title={The simplex algorithm is {NP}-mighty},
  author={Disser, Yann and Skutella, Martin},
  journal={ACM Transactions on Algorithms},
  volume={15},
  number={1},
  pages={1--19},
  year={2018},
  doi = {10.1145/3280847},
}

@inproceedings{huiberts2023upper,
  title={Upper and lower bounds on the smoothed complexity of the simplex method},
  author={Huiberts, Sophie and Lee, Yin Tat and Zhang, Xinzhi},
  booktitle={Proceedings of the 55th ACM Symposium on Theory of Computing (STOC)},
  pages={1904--1917},
  year={2023},
  doi={10.1145/3564246.3585124}
}

@article{smale1998mathematical,
  title={Mathematical problems for the next century},
  author={Smale, Steve},
  journal={The mathematical intelligencer},
  volume={20},
  number={2},
  pages={7--15},
  year={1998},
  publisher={Springer},
  doi={10.1007/BF03025291}
}

@article{santos2012counterexample,
  title={A counterexample to the {H}irsch conjecture},
  author={Santos, Francisco},
  journal={Annals of Mathematics},
  pages={383--412},
  year={2012},
  doi = {10.4007/annals.2012.176.1.7},
}

@book{fletcher2000practical,
  title={Practical Methods of Optimization},
  Author= {Roger Fletcher},
  year={2000},
  publisher={John Wiley \& Sons},
  doi = {10.1002/9781118723203}
}

@article{dantzig1963linear,
  title={Linear Programming and Extensions},
  author={Dantzig, George},
  journal={Princeton university press},
  year={1963},
  doi ={10.1515/9781400884179}
}

@article{kozlov1979polynomial,
title = {The polynomial solvability of convex quadratic programming},
journal = {USSR Computational Mathematics and Mathematical Physics},
volume = {20},
number = {5},
pages = {223-228},
year = {1980},
issn = {0041-5553},
doi = {https://doi.org/10.1016/0041-5553(80)90098-1},
author = {M.K. Kozlov and S.P. Tarasov and L.G. Khachiyan},
}

@inproceedings{karmarkar1984new,
  title={A new polynomial-time algorithm for linear programming},
  author={Karmarkar, Narendra},
  booktitle={Proceedings of the 16th ACM Symposium on Theory of Computing (STOC)},
  pages={302--311},
  year={1984},
  doi={10.1145/800057.808695}
}

@article{ye1989extension,
  title={An extension of {K}armarkar's projective algorithm for convex quadratic programming},
  author={Ye, Yinyu and Tse, Edison},
  journal={Mathematical Programming},
  volume={44},
  pages={157--179},
  year={1989},
  doi={10.1007/BF01587086}
}

@article{avis1978notes,
  title={Notes on {B}land’s pivoting rule},
  author={Avis, David and Chv{\'a}tal, Vasek},
  journal={Polyhedral Combinatorics: Dedicated to the memory of D.R.~Fulkerson},
  pages={24--34},
  year={1978},
  doi = {10.1007/BFb0121192}
}

@article{goldfarb1979worst,
  title={Worst case behavior of the steepest edge simplex method},
  author={Goldfarb, Donald and Sit, William Y.},
  journal={Discrete Applied Mathematics},
  volume={1},
  number={4},
  pages={277--285},
  year={1979},
  doi = {10.1016/0166-218X(79)90004-0},
}

@article{jeroslow1973simplex,
  title={The simplex algorithm with the pivot rule of maximizing criterion improvement},
  author={Jeroslow, Robert G.},
  journal={Discrete Mathematics},
  volume={4},
  number={4},
  pages={367--377},
  year={1973},
  doi = {10.1016/0012-365X(73)90171-4},
}

@article{klee1972good,
  title={How good is the simplex algorithm?},
  author={Klee, Victor and Minty, George J},
  journal={Inequalities},
  volume={3},
  number={3},
  pages={159--175},
  year={1972},
  url={https://archive.org/details/inequalities0000oved/page/n5/mode/2up}
}

@article{murty1980computational,
  title={Computational complexity of parametric linear programming},
  author={Murty, Katta G.},
  journal={Mathematical Programming},
  volume={19},
  number={1},
  pages={pp. 213--219},
  year={1980},
  doi={10.1007/BF01581642}
}

@inproceedings{friedmann2011subexponential_1,
  title={A subexponential lower bound for {Z}adeh’s pivoting rule for solving linear programs and games},
  author={Friedmann, Oliver},
  booktitle={Proceedings of the 15th Conference on Integer Programming and Combinatorial Optimization (IPCO)},
  pages={192--206},
  year={2011},
  doi = {10.1007/978-3-642-20807-2_16}
}

@inproceedings{friedmann2011subexponential_2,
  title={Subexponential lower bounds for randomized pivoting rules for the simplex algorithm},
  author={Friedmann, Oliver and Hansen, Thomas D. and Zwick, Uri},
  booktitle={Proceedings of the 43rd ACM Symposium on Theory of Computing (STOC)},
  pages={283--292},
  year={2011},
  doi = {10.1145/1993636.1993675},
}

@article{avis2017exponential,
  title={An exponential lower bound for {C}unningham’s rule},
  author={Avis, David and Friedmann, Oliver},
  journal={Mathematical Programming},
  volume={161},
  pages={pp. 271--305},
  year={2017},
  doi = {10.1007/s10107-016-1008-4},
}

@article{disser2023exponential,
  title={An exponential lower bound for {Z}adeh’s pivot rule},
  author={Disser, Yann and Friedmann, Oliver and Hopp, Alexander V.},
  journal={Mathematical Programming},
  volume={199},
  number={1-2},
  pages={865--936},
  year={2023},
  doi = {10.1007/s10107-022-01848-x},
}

@article{amenta1999deformed,
  title={Deformed products and maximal shadows of polytopes},
  author={Amenta, Nina and Ziegler, G\"unter M},
  journal={Contemporary Mathematics},
  volume={223},
  pages={pp. 57--90},
  year={1999},
  doi={10.1090/conm/223/03132}
}

@InProceedings{disser2026lower,
  author =	{Disser, Yann and Loho, Georg and Maat, Matthew and Mosis, Nils},
  title =	{{Lower Bounds for Ranking-Based Pivot Rules}},
  booktitle =	{43rd International Symposium on Theoretical Aspects of Computer Science (STACS)},
  pages =	{31:1--31:19},
  ISSN =	{1868-8969},
  year =	{2026},
  doi =	{10.4230/LIPIcs.STACS.2026.31}
}

@article{black2023normalizedweight,
  author       = {Alexander E. Black and
                  Jes{\'{u}}s A. De Loera and
                  Niklas L{\"{u}}tjeharms and
                  Raman Sanyal},
  title        = {The Polyhedral Geometry of Pivot Rules and Monotone Paths},
  journal      = {{SIAM} Journal on Applied Algebra and Geometry},
  volume       = {7},
  number       = {3},
  pages        = {623--650},
  year         = {2023},
  doi          = {10.1137/22M1475910},
}

@article{Dantzig82,
  author  = {George B. Dantzig},
  title   = {Reminiscences about the Origins of Linear Programming},
  journal = {Operations Research Letters},
  volume  = {1},
  number  = {2},
  pages   = {43--48},
  year    = {1982},
  doi     = {10.1016/0167-6377(82)90043-8}
}

@article{Kaibel11,
  author  = {Volker Kaibel},
  title   = {Extended Formulations in Combinatorial Optimization},
  journal = {Optima},
  volume  = {85},
  year    = {2011},
  pages   = {2--7},
  eprint  = {1104.1023},
  archivePrefix = {arXiv},
}

@article{Yannakakis91,
  author  = {Mihalis Yannakakis},
  title   = {Expressing combinatorial optimization problems by linear programs},
  journal = {Journal of Computer and System Sciences},
  volume  = {43},
  number  = {3},
  pages   = {441--466},
  year    = {1991},
  doi     = {10.1016/0022-0000(91)90024-Y}
}

@techreport{SCIP7,
  author = {Gerald Gamrath and Daniel Anderson and Ksenia Bestuzheva and Wei-Kun Chen and Leon Eifler and Maxime Gasse and Patrick Gemander and Ambros Gleixner and Leona Gottwald and Katrin Halbig and Gregor Hendel and Christopher Hojny and Thorsten Koch and Le Bodic, Pierre and Stephen J. Maher and Frederic Matter and Matthias Miltenberger and Erik M{\"u}hmer and Benjamin M{\"u}ller and Marc E. Pfetsch and Franziska Schl{\"o}sser and Felipe Serrano and Yuji Shinano and Christine Tawfik and Stefan Vigerske and Fabian Wegscheider and Dieter Weninger and Jakob Witzig},
  title = {{The SCIP Optimization Suite 7.0}},
  type = {Technical Report},
  institution = {Optimization Online},
  month = mar,
  year = {2020},
  url = {http://www.optimization-online.org/DB_HTML/2020/03/7705.html}
}

@article{Vielma2008,
  title = {A Lifted Linear Programming Branch-and-Bound Algorithm for Mixed-Integer Conic Quadratic Programs},
  volume = {20},
  ISSN = {1526-5528},
  DOI = {10.1287/ijoc.1070.0256},
  number = {3},
  journal = {INFORMS Journal on Computing},
  publisher = {Institute for Operations Research and the Management Sciences (INFORMS)},
  author = {Vielma,  Juan Pablo and Ahmed,  Shabbir and Nemhauser,  George L.},
  year = {2008},
  month = aug,
  pages = {438–450}
}

@article{VANDAELE2017217,
title = {On the linear extension complexity of regular n-gons},
journal = {Linear Algebra and its Applications},
volume = {521},
pages = {217-239},
year = {2017},
issn = {0024-3795},
doi = {10.1016/j.laa.2016.12.023},
author = {Arnaud Vandaele and Nicolas Gillis and François Glineur},
}

@article{Black27,
  author  = {Black, Alexander E. and Xue, Lei},
  title   = {Any Proof of Polynomial Hirsch Must be Completely Incoherent},
  journal = {To appear at SODA 2027},
  year    = {2026},
  eprint  = {2607.11628},
  archivePrefix = {arXiv},
}

@article{kaibel2013constructing,
  title={Constructing extended formulations from reflection relations},
  author={Kaibel, Volker and Pashkovich, Kanstantsin},
  journal={Facets of Combinatorial Optimization: Festschrift for Martin Gr{\"o}tschel},
  pages={77--100},
  year={2013},
  publisher={Springer},
  doi ={10.1007/978-3-642-38189-8_4}
}

@article{fiorini2012extended,
  title={Extended formulations for polygons},
  author={Fiorini, Samuel and Rothvo{\ss}, Thomas and Tiwary, Hans Raj},
  journal={Discrete \& computational geometry},
  volume={48},
  number={3},
  pages={658--668},
  year={2012},
  publisher={Springer},
  doi = {10.1007/s00454-012-9421-9}
}

@article{Khachiyan79,
  title={Polynomial algorithms in linear programming},
  author={Khachiyan, Leonid G.},
  journal={USSR Computational Mathematics and Mathematical Physics},
  volume={20},
  number={1},
  pages={53--72},
  year={1980},
  doi={10.1016/0041-5553(80)90061-0}
}

@techreport{glineur,
    author={Glineur, François},
    title={Polyhedral approximation of the second-order cone: computational experiments},
    institution={Faculté Polytechnique de Mons},
    year={2000},
    url={https://web.archive.org/web/20220505114903/https://perso.uclouvain.be/francois.glineur/oldwww/Papers/Image0001.ps.gz}
}

@article{barmann2016polyhedral,
  title={Polyhedral approximation of ellipsoidal uncertainty sets via extended formulations: a computational case study},
  author={B{\"a}rmann, Andreas and Heidt, Andreas and Martin, Alexander and Pokutta, Sebastian and Thurner, Christoph},
  journal={Computational Management Science},
  volume={13},
  pages={151--193},
  year={2016},
  publisher={Springer},
  doi = {10.1007/s10287-015-0243-0}
}

@article{cardenas2016,
  author={Angulo Cárdenas, Alejandro A. and Mancilla-David, Fernando and Palma-Behnke, Rodrigo E. and Espinoza, Daniel G.},
  journal={IEEE Transactions on Power Systems}, 
  title={A Polyhedral-Based Approach Applied to Quadratic Cost Curves in the Unit Commitment Problem}, 
  year={2016},
  volume={31},
  number={5},
  pages={3674-3683},

  doi={10.1109/TPWRS.2015.2499442}}

@article{cardinal2025inapproximability,
  author = {Cardinal, Jean and Steiner, Raphael},
  title = {Inapproximability of shortest paths on perfect matching polytopes},
  journal = {Mathematical Programming},
  volume = {210},
  number = {1},
  pages = {147--163},
  year = {2025},
  doi = {10.1007/s10107-023-02025-4}
}

@article{deloera2022pivot,
  author = {De Loera, Jesús A. and Kafer, Sean and Sanità, Laura},
  title = {Pivot rules for circuit-augmentation algorithms in linear optimization},
  journal = {SIAM Journal on Optimization},
  volume = {32},
  number = {3},
  pages = {2156--2179},
  year = {2022},
  doi = {10.1137/21M1419994}
}
\end{document}